\documentclass[11pt]{article}
\pdfoutput=1
\usepackage[ascii]{inputenc}
\usepackage[bookmarksnumbered, hypertexnames=false,colorlinks,colorlinks=true,linkcolor=blue,urlcolor=blue,citecolor=blue,anchorcolor=green]{hyperref}
\usepackage{bbm,braket,microtype,mathrsfs,amsmath,amssymb,color,amsthm,graphicx,enumitem,mathtools,ae,aecompl,float,cases,physics,appendix,verbatim,url,fullpage}
\usepackage[affil-it]{authblk}
\usepackage[all]{xy}
\usepackage[T1]{fontenc}
\usepackage[USenglish]{babel}
\usepackage{cleveref}
\usepackage{xcolor}
\usepackage{thmtools, thm-restate}
\usepackage{xpatch}
\usepackage{apptools}
\newtheorem{theorem}{Theorem}[section]

\newtheorem{proposition}[theorem]{Proposition}
\newtheorem{problem}[theorem]{Problem}
\newtheorem{corollary}[theorem]{Corollary}
\newtheorem{lemma}[theorem]{Lemma}

\newtheorem{remark}[theorem]{Remark}
\theoremstyle{definition}
\newtheorem{definition}[theorem]{Definition}
\floatstyle{boxed}\newfloat{Algorithm}{ht}{alg}
\crefname{theorem}{Theorem}{Theorems}
\crefname{question}{Question}{Questions}
\crefname{claim}{Claim}{Claims}
\crefname{proposition}{Proposition}{Propositions}
\crefname{problem}{Problem}{Problems}
\crefname{corollary}{Corollary}{Corollaries}
\crefname{lemma}{Lemma}{Lemmas}
\crefname{conjecture}{Conjecture}{Conjectures}
\crefname{openproblem}{Open Problem}{Open Problems}
\crefname{example}{Example}{Examples}
\crefname{fact}{Fact}{Facts}
\crefname{remark}{Remark}{Remarks}
\crefname{definition}{Definition}{Definitions}
\crefname{Algorithm}{Algorithm}{Algorithms}
\crefname{section}{Section}{Sections}
\crefname{appendix}{Appendix}{Appendices}
\numberwithin{equation}{section}
\DeclareMathOperator{\defi}{def}

\DeclareMathOperator{\sign}{sign}

\DeclareMathOperator{\capac}{cap}

\DeclareMathOperator{\Ten}{Ten}
\DeclareMathOperator{\poly}{poly}
\DeclareMathOperator{\expon}{exp}

\DeclareMathOperator{\rk}{rk}
\DeclareMathOperator{\GL}{GL}
\DeclareMathOperator{\SL}{SL}

\DeclareMathOperator{\diag}{diag}
\DeclareMathOperator{\supp}{supp}
\DeclareMathOperator{\ins}{instability}
\DeclareMathOperator{\ds}{ds}
\DeclareMathOperator{\dds}{dds}
\DeclareMathOperator{\sr}{slice-rank}
\DeclareMathOperator*{\argmin}{argmin}

\newcommand{\ot}{\otimes}

\newcommand{\eps}{\varepsilon}
\newcommand{\wt}[1]{\widetilde{#1}}
\newcommand{\RO}{\mathcal{R}_G}
\newcommand{\F}{{\mathbb{F}}}
\newcommand{\Q}{{\mathbb{Q}}}
\newcommand{\Z}{{\mathbb{Z}}}
\newcommand{\N}{{\mathbb{N}}}
\newcommand{\R}{{\mathbb{R}}}

\newcommand{\C}{{\mathbb{C}}}

\newcommand{\cH}{\mathcal H}

\newcommand{\cX}{\mathcal X}
\newcommand{\cR}{\mathcal R}

\allowdisplaybreaks
\makeatletter
\xpatchcmd{\thmt@restatable}
{\csname #2\@xa\endcsname\ifx\@nx#1\@nx\else[{#1}]\fi}
{\IfAppendix{\csname #2\@xa\endcsname}{\csname #2\@xa\endcsname\ifx\@nx#1\@nx\else[{#1}]\fi}}
{}{}
\makeatother

\begin{document}
\title{Alternating minimization, scaling algorithms, and the null-cone problem from invariant theory}
\author[1]{Peter B\"{u}rgisser}
\author[2]{Ankit Garg}
\author[3,4]{Rafael Oliveira}
\author[5,6]{Michael Walter}
\author[7]{Avi Wigderson}
\affil[1]{Institut f\"{u}r Mathematik, Technische Universit\"{a}t Berlin}
\affil[2]{Microsoft Research New England}
\affil[3]{Department of Computer Science, Princeton University}
\affil[4]{Department of Computer Science, University of Toronto}
\affil[5]{QuSoft, Korteweg-de Vries Institute for Mathematics, Institute of Physics, and Institute for Logic, Language and Computation, University of Amsterdam}
\affil[6]{Stanford Institute for Theoretical Physics, Stanford University}
\affil[7]{Institute for Advanced Study, Princeton}
\date{}
\maketitle

\begin{abstract}
Alternating minimization heuristics seek to solve a (difficult) global optimization task through iteratively solving a sequence of (much easier) local optimization tasks on different parts (or blocks) of the input parameters.
While popular and widely applicable, very few examples of this heuristic are rigorously shown to converge to optimality, and even fewer to do so efficiently.

In this paper we present a general framework which is amenable to rigorous analysis, and expose its applicability.
Its main feature is that the local optimization domains are each a group of invertible matrices, together naturally acting on tensors, and the optimization problem is minimizing the norm of an input tensor under this joint action.
The solution of this optimization problem captures a basic problem in Invariant Theory, called the \emph{null-cone problem}.

This algebraic framework turns out to encompass natural computational problems in combinatorial optimization, algebra, analysis, quantum information theory, and geometric complexity theory. It includes and extends to high dimensions the recent advances on (2-dimensional) \emph{operator scaling}~\cite{gurvits2004,GGOW,IQS15b}.

Our main result is a fully polynomial time approximation scheme for this general problem, which may be viewed as a multi-dimensional scaling algorithm. This directly leads to progress on some of the problems in the areas above, and a unified view of others. We explain how faster convergence of an algorithm for the same problem will allow resolving central open problems.

Our main techniques come from Invariant Theory, and include its rich \emph{non-commutative duality theory}, and new bounds on the bitsizes of coefficients of \emph{invariant polynomials}. They enrich the algorithmic toolbox of this very computational field of mathematics, and are directly related to some challenges in geometric complexity theory (GCT).
\end{abstract}
\clearpage\tableofcontents\clearpage

\section{Introduction and summary of results}

Alternating minimization refers to a large class of heuristics commonly used in optimization, information theory, statistics and machine learning.
It addresses optimization problems of the following general form.
Given a function
\begin{equation}\label{eq:block-form}
f\colon \cX_1\times\dots\times\cX_d \rightarrow \mathbb R
 \end{equation}
the goal is to find a global optimum
\begin{equation}\label{eq:minimization problem}
   \inf_{x_1\in\cX_1,\dots,x_d\in\cX_d} f(x_1, \ldots, x_d) .
\end{equation}
While both the function $f$ and its domain may be extremely complicated, in particular non-convex, the decomposition of the domain to $d$ {\em blocks} is such that the {\em local} optimization problems are all feasible.
More precisely, for every $i\in [d]$, and for every choice of $\alpha_j\in \cX_j$ with $j\neq i$, computing
\begin{equation*}
 \inf_{x_i\in\cX_i} \, f(\alpha_1,\ldots,\alpha_{i-1},x_i,\alpha_{i+1},\ldots,\alpha_d)
\end{equation*}
is easy.

A natural heuristic in such cases is to repeatedly make local improvements to different coordinates.
Namely, we start from some arbitrary vector $\alpha^{(0)}$, and generate a sequence $\alpha^{(0)}, \alpha^{(1)}, \dots, \alpha^{(t)}, \dots$, such that $\alpha^{(t+1)}$ is obtained by solving the above local optimization problem for some $i=i^{(t)}$, freeing up this variable while fixing all other coordinates according to $\alpha^{(t)}$.
The argmin of the optimum replaces the $i$th coordinate in $\alpha^{(t)}$ to create $\alpha^{(t+1)}$.

There is a vast number of situations which are captured by this natural framework.
For one example (we'll see more), consider the famous Lemke-Howson~\cite{Lemke-Howson} algorithm for finding a Nash-equilibrium in a 2-player game.
Here $d = 2$, and starting from an arbitrary pair of strategies, proceed by alternatingly finding a ``best response'' strategy for one player given the strategy of the other.
This local optimization is simply a linear program that can be efficiently computed.
As is well known, this algorithm always converges, but in some games it requires exponential time!

So, the basic questions which arise in some settings are under which conditions does this general heuristic converge at all, and even better, when does it converge efficiently.
Seemingly the first paper to give {\em provable} sufficient conditions (the ``5-point property'') for convergence in a \emph{general} setting was Csisz{\'a}r and Tusn{\'a}dy~\cite{CT84}.
An application is computing the distance (in their case, KL-divergence, but other metrics can be considered) between two convex sets in $\R^n$, as well as the two closest points achieving that distance.
Again, the local problem (fixing one point and finding the closest to it in the other convex set) is a simple convex program that has a simple efficient algorithm.
For two affine subspaces and the $\ell_2$-distance, von Neumann's alternating projection method~\cite{neumann1950functional} is an important special case with numerous applications.
As mentioned, in the past three decades numerous papers studied various variations and gave conditions for convergence, especially in cases where $d>2$ (often called also ``block-coordinate descent'') -- we cite only a few recent ones~\cite{XY13,RHL13,WYZ15} which address a variety of problems and have many references.
Much much fewer cases are known in which convergence is fast, namely requires only a polynomial number of local iterations.
Some examples include the recent works on matrix completion~\cite{JNS13,Hardt14}.
Our aim is to develop techniques that expand {\em efficient} convergence results in the following algebraic setting, which as we shall see is surprisingly general and has many motivations.

We will consider the case of minimizing~\eqref{eq:block-form} with very specific domain and function, that we in turn explain and then motivate (so please bear with this necessary notation).
First, each of the blocks~$\cX_i$ is a special
linear group, $\SL_{n_i}(\C)$ (which we will also abbreviate by $\SL(n_i)$), namely the invertible complex matrices of some size $n_i$
and determinant one.
Note that this domain is quite complex, as these groups are certainly not convex, and not even compact.
To explain the function $f$ to be minimized over this domain, consider the natural linear transformation (basis change) of the vector space $\C^{n_i}$ by matrices $A_i \in \SL(n_i)$.
Thus, the product group
$G := \SL(n_1) \times \SL(n_2) \times \cdots \times \SL(n_d)$
acts on tensors
$X \in V := \Ten(n_0,n_1,n_2,\dots,n_d) = \C^{n_0} \otimes \C^{n_1} \otimes \C^{n_2} \otimes \cdots \otimes \C^{n_d}$
of order $d+1$\footnote{To better capture past work and some mathematical motivations, we add a $0$-th component (on which there is no action at all) and make the tensor of order $d+1$. Equivalently, this means that we study tuples of tensors of order $d$.},
where (the basis change) $A_i$ is applied to all vectors (slices, ``fibers'') of $X$ along the $i$th dimension.
Now the objective function $f=f_X$~depends on an input tensor $X\in V$;
for any vector of matrices $A=(A_1,A_2, \dots A_d)\in G$, $f_X(A)$ is defined simply as the $\ell_2$-norm squared\footnote{Namely, the sum of squares of the absolute values of the entries.} of $A \cdot X$, the tensor resulting by the action of $A$ on~$X$.
Summarizing this description, for input $X\in V$ we would like to compute (or approximate) the quantity we call {\em capacity} of $X$, denoted $\capac(X)$ defined by
\begin{equation}\label{eq:alg-alt-min}
   \capac(X) := \inf_{A\in G}\, \|A\cdot X\|_2^2 .
\end{equation}
In particular, we would like to decide the {\em null-cone problem}: is the capacity of
$X$ zero or not?

While this formulation readily lends itself to an alternating minimization procedure (the local problems have an easily computed closed-form solution), the algorithms we will develop will be related to a dual optimization (scaling) problem that also has a similar form.
\emph{Our main result will be a fully polynomial time approximation scheme (FPTAS) for that dual problem!} This will also allow us to solve the null-cone problem.
However we defer this discussion and turn to motivations.

Now where do such problems and this notion of capacity naturally arise? Let us list a few sources, revealing how fundamental this framework really is.
Some of these connections go back over a century and some were discovered in the past couple of years.
Also, some are less direct than others, and some overlap.
Some of them, and others, will be elaborated on throughout the paper.
\begin{itemize}
\item {\bf Combinatorial Optimization.}
Computing matroid intersection (for linear matroids over $\Q$) easily reduces to checking zeroness of capacity for $d=2$ for very ``simple'' inputs $X$.
This was discovered in~\cite{gurvits2004}, who showed it remains true for an implicit version of this problem for which there was no known polynomial time algorithm.
His alternating minimization algorithm (called today  {\em operator scaling}) is efficient for such simple inputs.
It also generalizes similar algorithms for such problems as perfect matchings in bipartite graphs and maximum flow in networks, where again zeroness of capacity captures the decision versions of the problems.
The associated alternating minimization algorithm gives a very different and yet efficient way than is taught in algorithms texts for these classical problems (these go by the name of {\em matrix scaling}, originated in~\cite{Sink} and described explicitly in~\cite{LSW}).

\item {\bf Non-commutative algebra.}
The most basic computational problem in this field is the {\em word problem}, namely, given an expression over the free skew field (namely, a formula over $\Q$ in non-commutative variables), is it identically zero\footnote{This is an analog of the PIT problem in algebraic complexity, for non-commutative formulas with division.}.
As shown in~\cite{GGOW} this problem is reducible to testing if the capacity is 0, for $d=2$.
\cite{GGOW} significantly extends~\cite{gurvits2004}, proving that his algorithm is actually an FPTAS for {\em all} inputs, resolving the $d=2$ of our general problem above, and providing the first polynomial time algorithm for this basic word problem (the best previous algorithm, deterministic or probabilistic, required exponential time)\footnote{This problem was solved for finite characteristics by different methods in~\cite{IQS15b}.}.

\item {\bf Analysis.}
We move from efficiently testing {\em identities} to efficiently testing {\em inequalities}.
The Brascamp-Lieb inequalities~\cite{BL76,Lieb90} are an extremely broad class capturing many famous inequalities in geometry, probability, analysis and information theory (including H\"older, Loomis-Whitney, Shearer, Brunn-Minkowski, \dots).
A central theorem in this field~\cite{Lieb90} expresses feasibility, and the tightest possible constant if feasible, for every instantiation of such inequality in terms of capacity (defined somewhat differently, but directly related to ours).
In~\cite{GGOW2} a direct alternating minimization algorithm\footnote{Which can be viewed also as a reduction to the operator scaling algorithm in~\cite{GGOW}.} is given which efficiently approximates it on every input! Previous techniques to even bound capacity were completely ineffective, using compactness arguments.

\item {\bf Quantum Information Theory and Geometric Complexity Theory.}
We finally move to arbitrary $d$.
The {\em quantum marginal problem} (generalizing the classical {\em marginal problem} in probability theory) asks if a collection of marginals (namely, density matrices) on $d$ subsystems of a given quantum system are {\em consistent}, namely is there a density matrix on the whole system whose partial traces on the given subsystems result in these marginals~\cite{klyachko2006quantum}.
When the state of the whole system is pure and the marginals are proportional to identity matrices, each subsystem is maximally entangled with the others.
It is an important task to distill entanglement, that is, given a tensor, to transform it into such a \emph{locally maximally entangled} state by SLOCC\footnote{Acronym for \emph{Stochastic Local Operations and Classical Communication}; these are the natural actions when each of the subsystems is owned by a different party.}  operations~\cite{bennett2000exact,klyachko2002coherent}.
It turns out that this amounts precisely to the capacity optimization problem, and the matrices which minimize it solve this distillation problem.
One can thus view our FPTAS as achieving approximate entanglement distillation (see, e.g., \cite{VDB03,walter2013entanglement,walter2014multipartite}).
We note that for $d=3$, the non-zeroness of capacity captures a refinement of the {\em asymptotic positivity} of special {\em Kronecker coefficients}, a central notion in representation theory and geometric complexity theory (GCT) of~\cite{GCT1,buergisser2017membership,ikenmeyer2017vanishing}.
We refer to \cref{sec:quantum} for more detail.

\item {\bf Invariant Theory.}
We will be brief, and soon discuss this area at length, as it underlies most of the work in this paper.
Invariant theory studies symmetries, namely group actions on sets, and their invariants.
The action of our particular group $G$ above on the particular linear space of tensors $V$ is an example of the main object of study in invariant theory:
actions of reductive\footnote{A somewhat technical term which includes all classical linear groups.} groups $G$ on linear spaces $V$.
A central notion in the (geometric) study of such actions is the {\em null cone};
it consists of all elements in $v\in V$ which $G$ maps arbitrarily close to $0\in V$.
Here the connection to our problem is most direct:
the null cone of the action in our framework is {\em precisely} the input tensors $X$ whose capacity is 0!

We stress that previous algorithms for the null cone (for general reductive actions) required \emph{doubly} exponential time (e.g., \cite[Algorithm 4.6.7]{Sturmfels}). In contrast, our algorithm (for the general framework above) decides if a given tensor is in the null cone in \emph{singly} exponential time!
We will also discuss special cases where our FPTAS achieves polynomial time. The same null-cone problem is a direct challenge among the algorithmic problems suggested by Mulmuley towards the GCT program in~\cite{gct5}.
\end{itemize}

We now turn to discuss at more length the relevant features of invariant theory, explaining in particular the dual optimization problem that our algorithm actually solves.
Then we will state our main results.

\subsection{Invariant theory}

Invariant theory is an area of mathematics in which computational methods and efficient algorithms are extremely developed (and sought).
The reader is referred to the excellent texts~\cite{Derksen-Kemper, Sturmfels} for general background, focused on algorithms.
Using invariant theory was key for the recent development of efficient algorithms for operator scaling mentioned above~\cite{GGOW,IQS15b}, while these algorithms in turn solved basic problems in this field itself.
This paper proceeds to expand the applications of invariant theory for algorithms (and computational complexity!) and to improve the efficiency of algorithms for basic problems in the field.

Invariant theory is an extremely rich theory, starting with seminal works of Cayley~\cite{Cay46}\footnote{Which describes his ingenious algorithm, the Omega-process.}, to compute invariants of $\SL(n)$, of Hilbert~\cite{Hil90,Hil}\footnote{In which he proves his celebrated Nullstellensatz and Finite Basis theorem.} and others in the 19th century.
Again, we focus on linear actions of reductive groups $G$ on vector spaces $V$.
It turns out that in this linear setting the relevant invariants are polynomial functions of the variables (entries of vectors in $V$ with respect to some fixed basis), and the central problem is to understand the {\em invariant ring},
namely all polynomials which are left invariant (compute the same value) under the action of $G$.
Two familiar examples, which illustrate how good an understanding we can hope for, are the following:
\begin{itemize}
\item When $G=S_n$, the group of permutations of the coordinates of an $n$-vector in $V=\F^n$, the ring of invariants
are (naturally) all {\em symmetric polynomials}, which is (finitely!) generated by the $n$ elementary symmetric polynomials.
\item When $G=\SL(n)\times\SL(n)$\footnote{Here $\SL(n)$ is the group of invertible matrices of determinant 1.} acts on $n\times n$ matrices by changing basis of the rows and columns respectively, the ring of invariants
is generated by one polynomial in the entries of this matrix: the \emph{determinant}.
\end{itemize}
One of Hilbert's celebrated results in~\cite{Hil90,Hil} was that the invariant ring is {\em always} finitely generated.
Further understanding calls for listing generators, finding their algebraic relations, making them as low degree as possible, as few as possible, and as easily computable as possible.
Many new motivations for these arose from the GCT program mentioned above (and below).

The action of $G$ naturally carves the space $V$ into {\em orbits}, namely sets of the form $G\cdot v$ consisting of all images of $v\in V$ under the action of $G$.
Understanding when are two vectors in the same orbit captures numerous natural problems, e.g., {\em graph isomorphism} (when is a graph in the orbit of another, under vertex renaming), {\em module isomorphism} (when is a vector of matrices in the orbit of another, under simultaneous conjugation), and others.
Over the complex numbers $\C$ it is more natural to consider {\em orbit closures} and answer similar problems regarding membership of vectors in such orbit closures.
For an important example, the central approach of Mulmuley and Sohoni's GCT program~\cite{GCT1} to Valiant's conjecture that VP$\neq$VNP is to translate it to the question of whether the permanent is in the orbit closure of (a polynomially larger) determinant under a linear action on the matrix entries.
Understanding orbit closures is the subject of \emph{Geometric Invariant Theory}, starting with the seminal work of Mumford~\cite{Mum65} in the middle of the last century.
He proved that answering such questions is {\em equivalent} to computing invariants.
That is, the orbit closures of two vectors $v,w \in V$ intersect if and only if $p(v)=p(w)$ for all invariant polynomials $p$.

An important special case is to understand the orbit closures that contain the zero vector $0\in V$.
Note that by the theorem above, these are all vectors $v\in V$ for which $p(v)=0$ holds for all invariant polynomials $p$ without constant term.
This is an algebraic variety (nothing more than a zero set a system of polynomials), which is called the {\em null cone} of the action of $G$ on $V$.
Understanding the null cone is also central to geometric complexity theory from a topological viewpoint when {\em projectivizing} the space of orbits (making a point of each one, ``modding out'' by the group action).

A key to understanding the null cone is a duality theory for the group action, which may be viewed as a non-commutative analog of the duality theory for linear programming.
We will elaborate on this duality theory, essential to our work, in \cref{sec:scaling,sec:instability} and only sketch the essentials here.
It is clear that a vector $v\in V$ is in the null cone if its capacity is equal to zero, and so can be ``certified'' by exhibiting a sequence of group elements $g_t$ such that $\|g_t \cdot v\|$ approaches zero.
How can we ``certify'' that a vector $v\in V$ is {\em not} in the null cone?
The capacity formulation means that there must be some fixed $\delta>0$ such that, for every element $w =g\cdot v$ in the orbit of $v$, $\|w\|\geq \delta$.
Consider a point $w$ attaining this minimum distance to~$0$ (assume it exists).
There is a way to write down a non-commutative analog of the ``Lagrange multiplier'' conditions giving equations saying essentially that the derivative of the group action on $w$ in any direction is zero.
It turns out that the distance to satisfying these equations is {\em another, dual} optimization problem.
In particular, to ``certify'' non-membership of $v$ in the null cone, it suffices to exhibit a sequence of group elements $g_t$ such that $g_t\cdot v$ approaches distance 0 from satisfying these equations.
Again, we give plenty more detail on that in \cref{sec:scaling,sec:instability}.

What is remarkable is that in our setting, when $G$ is a product of $d$ groups as above, the new optimization has the exact same form as the original one we started with -- only with a different function to minimize!
Moreover, the set of equations decouples to $d$ subsets, and optimizing each {\em local} one is efficiently computable.
In short, this makes our new optimization problem again amenable to alternating minimization.
Further, the conditions that we need to satisfy may be viewed as ``scaling conditions'', which for $d=2$ can be shown to be equivalent to the matrix scaling conditions in the commutative case, and to the operator scaling conditions in the non-commutative case.
Thus, these scaling algorithms in combinatorial optimization to ``doubly stochastic'' position naturally arise from general considerations of duality in geometric invariant theory.
We will continue to borrow this name and call a tensor that satisfies the minimization conditions \emph{``$d$-stochastic"}, and we quantify how well a tensor~$Y$ satisfies $d$-stochasticity by a distance measure denoted~$\ds(Y)$.
For an input tensor~$X$, we then seek to compute the minimal distance to $d$-stochasticity, denoted $\dds(X)$, which is the infimum of $\ds(Y)$ over all $Y$ in the orbit of $X$ (and formally defined in~\cref{def:ds}).
Thus, summarizing, with $G$ and $V$ as before \cref{eq:alg-alt-min}, our new, dual optimization problem is
\begin{equation}\label{eq:dual-alt-min}
\dds(X) := \inf_{A\in G} \, \ds(A\cdot X).
\end{equation}
In our framework, $\ds(Y)$ has a very simple form.
Taking the quantum information theory view of the the problem (described above in the motivation), the tensor $Y$ captures a quantum system of $d$ local parts; then $\ds(Y)$ is simply the total $\ell_2$-distance squared of the $d$ subsystems to being maximally mixed (that is, proportional to normalized identity density matrices).

We now proceed to describe our main results:
the problems we address,
our algorithms for them,
and their complexities.
More formal statements of all will appear in the technical sections.

\subsection{Our results}

We start by describing our technical results, and then discussing the conceptual contribution of our paper.

\paragraph{Technical results.}

We fix $G$ and $V$ as above. While we work over the complex numbers, we will think of an input tensor $X \in V$ as an array of integers
(one can consider rational numbers, but this is no gain in generality) represented in binary.
The input size parameters will be the number of tensor entries $n=n_0 \times n_1 \times \dots \times n_d$,
and the maximum binary length of each entry, which will be denoted by $b$.
So the total input length may be thought of as $nb$.

Recall again that we have considered two dual optimization problems\footnote{While we do not know of any precise relationship between the value of these optimization problems, the vanishing behaviour of these two problems, as explained above, is dual to each other. $\capac(X) = 0$ iff $\dds(X) > 0$.} above for which the input is $X$:
\begin{equation}
\tag{\ref{eq:alg-alt-min}}
   \capac(X) = \inf_{A\in G}\, \|A\cdot X\|_2^2,
\end{equation}
\begin{equation}
\tag{\ref{eq:dual-alt-min}}
\dds(X) = \inf_{A\in G} \, \ds(A\cdot X).
\end{equation}

These lead to both exact and approximate computational problems for which we give new algorithms.

The exact {\em null-cone problem} we will consider is to test if the input tensor $X$ is in the null cone of the group action $G$.
As we discussed, $X$ is in the null cone iff $\capac(X) = 0$ and iff $\dds(X) >0$ (!).

\medskip

We will give two different {\em exponential} time algorithms for this problem. We note that the previous best algorithms (which work in the greater generality of all reductive group actions) required {\em doubly exponential} time! These algorithms may be very useful for the study of invariants for the actions of ``small'', specific groups\footnote{Invariant theory started with the study of invariants of binary, quadratic forms under the action of $\SL(2)$. Constructing complete sets of invariants for slightly larger groups is still open.}. We will discuss these algorithms soon.

The approximate problem we will consider will be to approximate $\dds(X)$\footnote{One could also consider the same for $\capac(X)$, but unlike $\dds(X)$, we have no faster algorithm for approximating $\capac(X)$ than computing it exactly.}.
Our approximation of $\dds$ runs in polynomial time in the input length $nb$ and the (additive) approximation $\epsilon$ -- thus a FPTAS is our main result\footnote{Some readers might object to our use of the term FPTAS for our algorithm since we have an additive approximation guarantee whereas the term is commonly used for multiplicative guarantees. However, note that a multiplicative guarantee for either $\capac(X)$ or $\dds(X)$ will need to determine their vanishing behaviour. This captures the null-cone problem, a polynomial time algorithm for which is the main open problem left open by this work. Also note that we don't have any guarantee on approximating $\dds(X)$ if $\dds(X) > 0$, which is also left as an interesting open problem. However, approximating $\dds(X)$ when it is $0$ is more fundamental because of the connection to the null cone.}, and we state it first.

\begin{restatable}[Main theorem]{theorem}{MainTheorem}\label{thm:main}
There is a $\poly(n, b, \frac1\eps)$ time deterministic algorithm (\cref{alg:scaling}) that, given a tensor
$X \in V$ with integer coordinates of bit size bounded by~$b$,
either identifies that $X$ is in the null cone or outputs a ``scaling'' $Y \in G \cdot X$ such that $\ds(Y)<\eps$.
\end{restatable}

We present the algorithm right below in \cref{subsec:scaling algo} and analyze it in \cref{subsec:scaling analysis}.
It is a realization of the alternating minimization scheme discussed before, and may be viewed as a \emph{scaling algorithm}.
Indeed, it generalizes the FPTAS given in~\cite{GGOW} for $d=2$, which is the operator scaling case.
While the analysis is similar in spirit, and borrows many ideas from~\cite{GGOW}, we note that for higher dimensions $d>2$, the understanding of invariant polynomials was much poorer.
Thus, new ideas are needed, and in particular we give explicit\footnote{The generators are explicitly defined but may not be efficiently computable by algebraic circuits.} generators (via classical Invariant Theory techniques) of the invariant ring (for all degrees), all of which have small coefficients, and we bound the running time in a way which is {\em independent} of degree bounds.

Moving to the algorithms for the exact null-cone problem, as discussed, we have two exponential time ones, improving the known doubly exponential ones.

The first algorithm is a natural instantiation of the FPTAS above, with the right choice of $\epsilon$.
Specifically, we prove that when if $X$ is not in the null cone, $\dds(X) \geq \exp(-O(n \log n))$.
So picking $\epsilon = \exp(-O(n\log n))$ suffices.
This algorithm is described in \cref{thm:null cone exp scaling}.

Our second algorithm is much more general, and applies to the action of a product of general linear groups on {\em any} vector space $V$, not just the the set of tensors.
Here, instead of giving explicit generators, we simply prove their existence, via Cayley's famous {\em Omega-process}. Namely, we prove that in this very general case there are always generating invariants which have both exponential degree and whose coefficients have exponential bit length. These bounds turns out to be sufficient to carry out our time analysis.
The bound on the size of coefficients is (to the best of our knowledge) a new structural result in Invariant Theory, and what our analysis demonstrates is its usefulness for solving computational problems.

The corollaries of these main results to specific areas, which we described in some of the motivation bullets earlier in the introduction, are presented in the different technical sections.

\paragraph{Conceptual contributions.}

We believe that the general framework we present, namely an \emph{algebraic framework of alternating minimization}, establishes an important connection between the fields of optimization and invariant theory.
As we saw, this framework is very general and encompasses many natural problems in several fields.
It exposes the importance of {\em symmetry} in many of these problems, and so invites the use of tools from invariant theory which studies symmetry to the analysis of algorithms.
This is akin to the GCT program, which connects algebraic complexity theory with invariant theory and representation theory, exposing the importance of using symmetry to proving lower bounds.

At the same time we expose basic computational problems of invariant theory, which are in great need of better algorithms and algorithmic techniques from computer science and optimization to solve them.
The advent of operating scaling already pointed out to this power in their utility of alternate minimization and scaling {\em numerical} algorithms to a field in which most work was {\em symbolic}, and we expand on this theme in this paper.
But there are many other optimization methods which can be useful there, and we note only one other example, that may in particular be useful for the very null-cone problem we study.
While the problem of capacity optimization is not convex under the usual Euclidean metric, it is actually {\em geodesically convex}~\cite{KN79,NM84,Woodward11} if one moves to the natural (hyperbolic) metric on the group. This opens up this problem (and related ones) to avenues of attack from techniques of classical optimization (gradient and other descent methods, interior point methods, etc.) when generalized and adapted to such metrics. We believe that the area of
geodesic convex optimization, which itself in its nascent stages (from an algorithmic standpoint), is likely to have a significant impact on algorithmic invariant theory. See more on algorithmic techniques for geodesic convexity in the work of~\cite{Zhang_Sra16} and the references therein.

\subsection{Plan of the paper}

In the remainder of the paper, we explain the problems that we study in more detail; we discuss our results and techniques and point out various connections and interesting open problems.
In~\cref{sec:nullcone as optimization}, we expand on the null-cone problem and motivate a natural alternating minimization algorithm for it.
In~\cref{sec:scaling}, we explain the Kempf-Ness theorem, which can be seen as a form of noncommutative duality, and we present our scaling algorithm for solving the null cone problem and its analysis.
We first give an overview of the techniques that go into its analysis.
We then construct a generating set of invariants via Schur-Weyl duality and lower bound the capacity of tensors that are not in the null cone.
Lastly, we give a precise mathematical analysis of our algorithm.
In~\cref{sec:instability}, we describe the Hilbert-Mumford criterion for the null cone and quantify the instability of a tensor in the null cone, which gives a way of characterizing when our algorithm can decide the null-cone problem in polynomial time.
\Cref{sec:quantum} explores connections to the quantum marginal problem and Kronecker coefficients and \cref{sec:slice rank} connects the null cone to the notion of slice rank of tensors.
In \Cref{sec:cayley}, we describe an explicit construction of Reynolds operators via Cayley's Omega process, and we present alternative, weaker bounds on the capacity and a resulting algebraic exponential-time algorithm for the null-cone problem, which generalizes to arbitrary polynomial actions of products of special linear groups.
In~\cref{sec:conclusion}, we mention some open problems.
To make the paper self-contained, we supply concrete proofs of well-known statements from geometric invariant theory in \cref{app:git}.

\section{Null-cone problem as an optimization problem}\label{sec:nullcone as optimization}

In the introduction, we defined the \emph{null cone} as the set of tensors~$X$ such that the zero vector $0\in V$ lies in the orbit closure $\overline{G\cdot X}$, i.e., there exists a sequence of group elements $A^{(1)},A^{(2)},\dots$ sending the vector to zero, $\lim_{j \rightarrow \infty} A^{(j)} \cdot X = 0$.
Thus the \emph{null-cone problem} amounts to the optimization problem~\eqref{eq:alg-alt-min}:
A tensor $X\in V = \Ten(n_0,n_1,\dots,n_d)$ is in the null cone if and only if the \emph{capacity} is equal to zero, namely:
\begin{equation*}
   \capac(X) = \inf_{A\in G}\, \norm{A\cdot X}_2^2 = \min_{Y \in \overline{G \cdot X}} \norm{Y}_2^2 = 0.
\end{equation*}

\begin{remark}
We are denoting the above optimization problem by $\capac(X)$, short for \emph{capacity}, to remain consistent with similar notions defined in previous papers like \emph{\cite{GurYianilos,gurvits2004,gur_Waerden,GGOW}}.
In most of these cases, the capacity notion that they consider is also looking at the optimization problem: ``minimum $\ell_2$-norm in an orbit closure" for specific group actions.
As phrased in these papers, they might look different (e.g., they involve determinants in place of $\ell_2$ norms) but there is a tight connection between the two via the AM-GM inequality (i.e., the inequality of arithmetic and geometric means) in one direction and a single alternating minimization step in the other direction.
\end{remark}

A fundamental theorem of Hilbert~\cite{Hil} and Mumford~\cite{Mum65} gives an alternative characterization of the null cone.
It states that the null cone is precisely the set of tensors $X$ on which all invariant polynomials~$P$ (without constant term) vanish.
This is the starting point to the field of Geometric Invariant Theory.

\label{page:qualitative}
It is instructive to see why $P(X)=0$ must hold for any tensor~$X$ in the null cone:
Since $P$ is $G$-invariant, $P(Y)=P(X)$ for all $Y$ in the $G$-orbit of $X$, and, by continuity, also for all $Y$ in the closure of the orbit.
If the tensor~$X$ is in the null cone, then $Y=0\in V$ is in its closure.
But then $P(X)=P(0)=0$, since the polynomial~$P$ has no constant term.
In \cref{prp:inf norm on orbit} we derive a more subtle, \emph{quantitative} version of this observation.
It will be a fundamental ingredient to the analysis of our algorithm.

\medskip

Since the group $G = \SL(n_1)\times\cdots\times\SL(n_d)$ consists of tuples $A = (A_1,\dots,A_d)$, an alternating minimization algorithm suggests itself for solving the optimization problem $\capac(X)$:
Fix an index~$i\in[d]$ and optimize only over a single $A_i$, leaving the $(A_j)_{j\neq i}$ unchanged.
This gives rise to an optimization problem of the following form:
$$
\inf_{A_S \in \SL(n_S)} \norm{\left( A_S \otimes I_{n_T} \right) \cdot Y}_2^2
$$
Here $n_S = n_i$, $n_T = \prod_{j \neq i} n_j$ and $Y = A \cdot X$.
This optimization problem has a closed form solution in terms of the following quantum mechanic analog of the marginals of a probability distribution:

\begin{definition}[Quantum marginal]\label{def:quantum marginal}
Given a matrix $\rho$ acting on $\cH_S \otimes \cH_T$, where $\cH_S$ and $\cH_T$ are Hilbert spaces of dimensions $n_S$ and $n_T$, respectively (e.g., $\cH_S = \C^{n_S}$ and $\cH_T = \C^{n_T}$).
Then there is a unique matrix $\rho_S$ acting on $\cH_S$, called the \emph{quantum marginal} (or reduced density matrix) of $S$, such that
\begin{equation}\label{eq:rdm}
\tr \left[\left( M_S \otimes I_T\right) \rho\right] = \tr \left[ M_S \rho_S\right]
\end{equation}
for every $M_S$ acting on $\cH_S$.
\end{definition}

This point of view gives rise to an important interpretation of our result (\cref{sec:quantum}).

It can be easily seen that if $\rho$ is positive semidefinite (PSD) then $\rho_S$ is PSD, and that it has the same trace as $\rho$.
It is given by the following explicit formula
\[
  \rho_S = \sum_{k=1}^{n_T} \left(I_S \otimes e_k\right)^{\dagger}
  \rho \left(I_S \otimes e_k\right) .
\]
Here the $e_k$'s are the elementary column vectors of dimension $n_T$.

Now observe that\footnote{Note that we can also view tensors in $\Ten(n_0,n_1,\dots,n_d)$ as column vectors in the space $\C^{n_0} \otimes \cdots \otimes \C^{n_d}$.}
\[
\left\| \left( A_S \otimes I_T \right) \cdot Y\right\|_2^2 = \tr \left[ \left(A_S^{\dagger} A_S \otimes I_T\right) Y Y^{\dagger}\right].
\]
Let $\rho_i$ be the partial trace of $\rho = Y Y^{\dagger}$ obtained by tracing out all the Hilbert spaces except the $i^{\text{th}}$ one, and rename $A_i = A_S$. Then, using \cref{eq:rdm},
\[
\tr \left[ \left(A_S^{\dagger} A_S \otimes I\right) Y Y^{\dagger}\right] = \tr \left[ A_i^{\dagger} A_i \rho_i \right].
\]
Hence we are left with the optimization problem:
\begin{equation}\label{eq:single optimization}
\inf_{A_i \in \SL(n_i)} \tr \left[ A_i A_i^{\dagger} \rho_i \right] .
\end{equation}
We can see by the AM-GM inequality that the optimum of the program~\eqref{eq:single optimization} is $n_i \det(\rho_i)^{1/n_i}$, which is achieved for $A_i = \det(\rho_i)^{1/2n_i} \rho_i^{-1/2}$ (if $\rho_i$ is not invertible then we define the inverse on the support of $\rho_i$; the infimum will at any rate be zero).

We thus obtain the core of an alternating minimization algorithm for the capacity $\capac(X)$.
At each step $t$, select an index $i\in[d]$, compute the quantum marginal $\rho_i$ of the current tensor $X^{(t)}$, and update by performing the following \emph{scaling step}:
\begin{equation}\label{eq:scaling step}
  X^{(t+1)} \leftarrow \det(\rho_i)^{1/2n_i} \rho_i^{-1/2} \cdot X^{(t)}.
\end{equation}
Here and in the following we use the abbreviation
\[ A_i \cdot Y := (I_{n_0}\ot I_{n_1}\ot\dots\ot I_{n_{i-1}}\ot A_i\ot
  I_{n_{i+1}}\ot\dots\ot I_{n_d}) \cdot Y , \]
where we act nontrivially on the $i$-th tensor factor only.

We will analyze essentially this algorithm, augmented by an appropriate method for selecting the index $i$ at each step and a stopping criterion (see \cref{alg:scaling} in the next section).

\section{Noncommutative duality and the tensor scaling algorithm}\label{sec:scaling}

To obtain a scaling algorithm with rigorous guarantees, we will now derive a dual formulation of the optimization problem $\capac(X)$.
This will achieved by a theorem from Geometric Invariant Theory, due to Kempf and Ness, which can be understood as part of a noncommutative duality theory.

We briefly return to the setting of a general action of a group $G$ on a vector space $V$ to explain this result.
Fix a vector $v \in V$ and consider the function $f_v(g) := \norm{g\cdot v}_2^2$.
The \emph{moment map} at $v$, $\mu(v)$, measures how much $f_v(g)$ changes when we perturb $g$ around the identity.
So $\mu(v)$ is just the derivative of the function $f_v(g)$ at the identity element (in a precise sense which we do not need to define here).
Now suppose $v$ is not in the null cone.
Then there exists a \emph{nonzero} vector $w \in \overline{G\cdot v}$ of minimal norm in the orbit closure.
Since the norm is minimal, perturbing $w$ by the action of group elements $g$ close to the identity element does not change the norm to first order.
Hence $\mu(w) = 0$.
To summarize, if $v$ is not in the null cone, then there exists $0 \neq w \in \overline{G\cdot v}$ such that $\mu(w) = 0$.
This condition can be understood as a vast generalization of the notion of ``doubly stochastic'', which one might be tempted to call \emph{``$G$-stochastic''}.

A remarkable theorem due to Kempf and Ness~\cite{KN79} asserts that this is not only a necessary but in fact also sufficient condition, i.e., $v$ is not in the null cone if and only if there exists $0 \neq w \in \overline{G\cdot v}$ such that $\mu(w) = 0$.
This is a \emph{duality theorem} and (along with the Hilbert-Mumford criterion discussed below) should be thought of as a non-commutative version of linear programming duality, which arises when the group $G$ is commutative (see \cref{rem:duality} below, \cite{sury00}, and the proof of \cref{thm:kempf ness} in \cref{app:git}).
For a detailed discussion of moment maps, we refer the reader to~\cite{Woodward11,walter2014multipartite}.

\bigskip
We now return to our setting where $G = \SL(n_1) \times \cdots \times \SL(n_{d})$ acts on tensors in $V = \Ten(n_0,n_1,\dots,n_d)$.
We first define the notion of a scaling and then instantiate the Kempf-Ness theorem in the situation at hand:

\begin{definition}[Scaling]
A tensor $Y \in \Ten(n_0,n_1,\ldots, n_d)$ is called a (\emph{tensor})  \emph{scaling} of
$X \in \Ten(n_0,n_1\ldots,n_d)$ if $Y \in G \cdot X$.
\end{definition}

This notion of scaling generalizes the notions of operator scaling~\cite{gurvits2004,GGOW} and matrix scaling~\cite{Sink}.
It is immediate from the definition of the capacity that a tensor $X$ is in the null cone if and only if any of its scalings $Y$ lies in the null cone.

We now state the Kempf-Ness theorem for tensors:

\begin{restatable}[Duality, special case of~\cite{KN79}]{theorem}{KempfNess}\label{thm:kempf ness}
A tensor $X \in \Ten(n_0,\ldots,n_d)$ is \emph{not} in the null cone if and only if there exists $Y \in \overline{G \cdot X}$ such that the quantum marginals $\rho_1,\dots,\rho_d$ of $\rho=YY^\dagger/Y^\dagger Y$ of the last $d$ tensor factors are given by $\rho_i = I_{n_i} / n_i$.
\end{restatable}

\noindent To be self-contained, we give a proof of \cref{thm:kempf ness} in \cref{app:git}.

We note that the trace of $\rho=YY^\dagger/Y^\dagger Y$ is one.
This normalization is very natural from the quantum viewpoint.
In quantum theory, PSD matrices of unit trace describe quantum states of multi-partite systems and the quantum marginals describe the state of subsystems.
The condition that $\rho_i$ is proportional to the identity means that the state of the $k^\text{th}$ subsystem is maximally mixed or, for $\rho$ as above, that the $k^\text{th}$ system is maximally entangled with the rest.
We discuss applications of our result to quantum information theory in \cref{sec:quantum}.

\medskip

The condition that each $\rho_i$ is proportional to the identity matrix generalizes the notion of ``doubly stochastic'', which arises in the study of the left-right action/operator scaling~\cite{gurvits2004,GGOW}.
We will refer to it as \emph{``$d$-stochastic''}, which we define formally below.

\begin{definition}[$d$-stochastic]
  A tensor $Y \in \Ten(n_0, \ldots, n_d)$ is called \emph{$d$-stochastic} if the quantum marginals of $\rho=YY^\dagger/Y^\dagger Y$ of the last $d$ tensor factors are given by $\rho_i = I_{n_i} / n_i$ for $i\in[d]$.
\end{definition}

More explicitly, we can state the condition that $\rho_i=I_{n_i}/n_i$ as follows:
We want that the slices of the tensor $Y$ in the $k^{\text{th}}$ direction are pairwise orthogonal and that their norm squares are equal to $1/n_i$.
That is, for each $i\in[d]$ consider the order $d$ tensors $Y^{(1)}, \ldots, Y^{(n_i)} \in \Ten(n_0,\dots,n_{i-1},n_{i+1},\dots,n_d)$ defined as
$$
Y^{(j)}(j_0,\ldots, j_{i-1}, j_{i+1},\ldots, j_d) := Y(j_0, \ldots, j_{i-1}, j, j_{i+1}, \ldots, j_d);
$$
then the following should hold:
$$
\langle Y^{(j)}, Y^{(j')} \rangle = \frac{1}{n_i} \delta_{j,j'} .
$$
Here $\delta_{i,j}$ is the Kronecker delta function and the inner product is the Euclidean inner product of tensors.

\medskip

\Cref{thm:kempf ness} implies that $X$ is \emph{not} in the null cone if and only if we can find scalings $Y$ whose last $d$ quantum marginals are arbitrarily close to $I_{n_i}/n_i$.
We can therefore rephrase the null-cone problem as \emph{another, dual} optimization problem, where we seek to minimize the distance of the marginals to being maximally mixed.
This is captured by the following definitions:

\begin{definition}[Distance to $d$-stochasticity]\label{def:ds}
Let $Y \in \Ten(n_0,n_1,\ldots,n_d)$ be a tensor and $\rho=YY^\dagger/Y^\dagger Y$, with quantum marginals $\rho_1,\dots,\rho_{d}$ on the last $d$ systems.
Then we define the \emph{distance to $d$-stochasticity} $\ds(Y)$ as the (squared) distance between the marginals and $d$-stochastic marginals,
\[
\ds(Y) := \sum_{i=1}^d \norm{\rho_i - \frac{I_{n_i}}{n_i}}_F^2,
\]
where $\lVert M \rVert_F := (\tr M^\dagger M)^{1/2}$ is the Frobenius norm.
Following \cref{eq:dual-alt-min}, we further define the \emph{minimal distance to $d$-stochasticity} as
\[
\dds(X) := \inf_{A\in G} \, \ds(A\cdot X).
\]
\end{definition}

Using this language, \cref{thm:kempf ness} states that $X$ is in the null cone if and only if the minimal distance to $d$-stochasticity is nonzero.
We summarize the duality between the two optimization problems:
\begin{equation}\label{eq:nc duality}
\text {$X$ is in the null cone}
\quad\Longleftrightarrow\quad
\capac(X) = 0
\quad\Longleftrightarrow\quad
\dds(X) > 0
\end{equation}

\begin{remark}\label{rem:duality}
According to \cref{eq:nc duality}, for any tensor~$X$ \emph{exactly one} of the following two statements is true:
$X$ is in the null cone ($\capac(X)=0$) or
its orbit closure contains a $d$-stochastic tensor ($\dds(X)=0$).
Such dichotomies are well-known from the duality theory of linear programming (Farkas' lemma, Gordan's lemma, the duality between membership in a convex hull and the existence of separating hyperplanes, etc.).
In the case of the (commutative) group of diagonal matrices, one recovers precisely these linear programming dualities from the general noncommutative framework.
\end{remark}

\subsection{The tensor scaling algorithm}\label{subsec:scaling algo}

If $X$ is \emph{not} in the null cone then there exist scalings~$Y$ such that $\ds(Y)$ is arbitrarily small.
The main technical result of this paper is an algorithm that finds such scalings for any fixed $\eps>0$.
It is a generalization of the results obtained for matrix and operator scaling in~\cite{LSW,GurYianilos,gurvits2004,GGOW}.
Recall that we use~$n$ to denote the total number of coordinates of the tensor, $n = n_0 \cdots n_d$.

\begin{restatable}{theorem}{MainTheoremPrecise}\label{thm:main precise}
Let $X \in \Ten(n_0,\dots,n_d)$ be a (nonzero) tensor whose entries are integers of bitsize no more than $b$.
Let $\ell = \min_{i \ge 1} n_i$ and suppose that $\varepsilon \le d/(\max_{i\geq1} n_i^2)$.
Then \cref{alg:scaling} with $T \geq \frac {18\ln2} {\ell \eps} d \left(  b + \log n \right)$ iterations either identifies that $X$ is in the null cone or outputs a scaling $Y \in G\cdot X$ such that $\ds(Y)\leq\varepsilon$.
\end{restatable}

To motivate \cref{alg:scaling}, note that, for every $i\in[d]$, we would like the quantum marginal of the $i^\text{th}$ system to be proportional to the identity matrix.
Suppose the quantum marginal on the $i^{\text{th}}$ system is $\rho_i$.
Then by acting on this system by any $A_i\in\SL(n_i)$ proportional to $\rho_i^{-1/2}$ we can precisely arrange for the quantum marginal on $i^{\text{th}}$ system to be proportional to the identity matrix\footnote{It is not hard to see that if $X$ is not in the null cone then $\rho_i$ is invertible (\cref{lem:singular marginal implies null cone}).}.
(But this might disturb the quantum marginals on the other systems!)
An appropriate choice of $A_i$ that ensures that it has determinant one is $A_i = \det(\rho_i)^{1/2n_i} \rho_i^{-1/2}$.
Thus we find that the alternating minimization heuristics in \cref{eq:scaling step} that we previously derived for $\capac(X)$ is equally natural from the perspective of the dual problem of minimizing $\dds(X)$!
Remarkably, iterating this operation in an appropriate order of subsystems $i$ does ensure that \emph{all} the quantum marginals get arbitrarily close to the normalized identity matrices -- provided that $X$ is not in the null cone.
(If $X$ is in the null cone, then it will lead to a sequence of scalings of norm arbitrarily close to zero, certifying that $X$ is in the null cone.)

\begin{Algorithm}
\textbf{Input}: A tensor $X$ in $\Ten(n_0,n_1,\ldots,n_d)$ with integer entries (specified as a list
of entries, each encoded in binary, with bit size $\le b$). \\[.3ex]

\textbf{Output}: Either the algorithm correctly identifies that $X$ is in the null cone, or it outputs a scaling $Y$ of $X$ such that $\ds(Y) \leq \eps$. \\[.3ex]

\textbf{Algorithm}:
\begin{enumerate}
\item\label{it:scaling step 1} If any of the quantum marginals of $X$ is singular, then output that the tensor $X$ is in the null cone and return.
Otherwise set $Y^{(1)} = X$ and proceed to step~\ref{it:scaling step 2}.
\item\label{it:scaling step 2} For $t=1,\dots,T = \poly(n, b , 1/\eps)$, repeat the following:
\begin{itemize}
\item Compute the quantum marginals $\rho_i$ of $Y^{(t)} Y^{(t),\dagger} / Y^{(t),\dagger} Y^{(t)}$ and find the index $i\in[d]$ for which $\norm{\rho_i - \frac {I_{n_i}} {n_i}}_F^2$ is largest.
If $\norm{\rho_i - \frac {I_{n_i}} {n_i}}_F^2 < \eps/d$, output $Y^{(t)}$ and return.
\item Otherwise, set $Y^{(t+1)} \leftarrow \det(\rho_i)^{1/2n_i} \rho_i^{-1/2} \cdot Y^{(t)}$.
\end{itemize}
\item\label{it:scaling step 3} Output that the tensor $X$ is in the null cone.
\end{enumerate}
\caption{Scaling algorithm for the null-cone problem}\label{alg:scaling}
\end{Algorithm}

\begin{remark}[Rational Entries]
Notice that we required our input tensor $X$ to have integer entries.
In general, an input tensor $X$ could have rational entries.
If this is the case, we can simply multiply $X$ by the denominators to obtain an integral tensor $X'$, on which we can then apply our algorithm above.
Since multiplying a tensor by a nonzero number does not change the property of being in the null cone,
our algorithm will still be correct.
The following remarks discuss the changes in bit complexity of our algorithm.
In this case, the size of $X'$ is at most $b \cdot n$, and and therefore the bound for the number of iterations $T$ is still $\poly(n, b, 1/\eps)$, as we wanted.
\end{remark}

When analyzing bit complexity, \Cref{alg:scaling} is not directly applicable, since it computes inverse square roots of matrices.
Even if the entries of the square roots were always rational, the algorithm would also iteratively multiply rationals with very high bit complexity.
We can handle these numerical issues by truncating the entries of the
scalings $Y^{(t)}$ of $X$, as well as the quantum marginals $\rho_k$,
to $\poly(n, b, \frac{1}{\eps})$ many bits after the decimal point.
Then, in a similar way as done in~\cite{gurvits2004,GGOW}, we can prove that there is essentially no change in the convergence required in the number of iterations $T$.
Since each arithmetic operation will be done with numbers of $\poly(n, b, \frac1\eps)$ many bits, this truncation ensures that we run in polynomial time.
As a consequence, we obtain our main theorem, which we had already announced in the introduction:

\MainTheorem*

We remark that for $d=2$, our algorithm in essence reduces to the algorithm of~\cite{GurYianilos,gurvits2004,gur_Waerden,GGOW}; for $n_0=1$, it refines the algorithm proposed previously in~\cite{VDB03} without a complexity analysis.
Indeed, \Cref{thm:main} is a generalization of the main theorem of~\cite{GGOW} on operator scaling (or the left-right action).
There it was enough to take $\eps = 1/n$ to be sure that the starting tensor is not in the null cone.

In our much more general setting it is still true that there exists some $\eps=\eps(n)>0$ such that $\ds(Y)>\eps$ for any $Y$ in the null cone.
Unfortunately, in our general setting, $\eps$ will be exponentially small (see \cref{lem:gap} in the next section).
Therefore, \cref{thm:main} with this choice of $\eps$ gives an exponential time algorithm for deciding the null-cone problem:

\begin{theorem}[Null-cone problem]\label{thm:null cone exp scaling}
There is a $\poly(b)\exp(O(n \log n))$
time deterministic algorithm
$($\cref{alg:scaling} with $\eps=\exp(-O(n\log n)))$ that,
given a tensor $X \in \Ten(n_0,n_1,\ldots,n_d)$ with integer
coordinates of bit size at most~$b$,
decides whether $X$ is in the null cone.
\end{theorem}

Thus our algorithm does not solve the null-cone problem in polynomial time for general tensor actions, and this remains an excellent open question!

Nevertheless, \cref{thm:main} as such, with its promise that $\ds(Y)<\eps$, is of interest beyond merely solving the null-cone problem.
It has applications to quantitative notions of instability in geometric invariant theory, to a form of multipartite entanglement distillation in quantum information theory, and to questions about the slice rank of tensors, which underpinned recent breakthroughs in additive combinatorics.
We will discuss these in \cref{sec:instability,sec:quantum,sec:slice rank} below.

\subsection{Analysis sketch}\label{subsec:scaling analysis}

To analyze our algorithm and prove \cref{thm:main precise}, we follow a three-step argument similar to the analysis of the algorithms for matrix scaling and operator scaling in~\cite{LSW,gurvits2004,GGOW} once we have identified the appropriate potential function and distance measure.
In the following we sketch the main ideas and we refer to \cref{subsec:invariants,subsec:detailed analysis algo} for all technical details.

As potential function, we will use the \emph{norm squared} of the tensor, denoted $f(Y)=\norm{Y}_2^2$, and the distance measure is $\ds(Y)$ defined above in \cref{def:ds}.
Note that these the two functions are exactly dual to each other in our noncommutative optimization framework (see \cref{eq:nc duality,eq:alg-alt-min,eq:dual-alt-min})!

\medskip

The following two properties of the capacity will be crucial for the analysis:
\begin{enumerate}
\item[A.] $\capac(X) \ge 1/(n_1\cdots n_{d})^2$ if $X \in \Ten(n_0,\ldots,n_d)$ is a tensor with integral entries that is not in the null cone (\cref{prp:inf norm on orbit}).
\item[B.] $\capac(X)$ is invariant under the action of the group $ \SL(n_1) \times \cdots \times \SL(n_{d})$.
\end{enumerate}
Using the above properties, a three-step analysis follows the following outline:
\begin{enumerate}
\item An upper bound on $f(Y^{(1)}) = f(X)$ from the input size.
\item As long as $\ds(Y)\geq\eps$, the norm squared decreases by a factor $2^{-\frac{\ell \eps}{6 d \ln2}}$ in each iteration:
$f(Y^{(t+1)}) \leq 2^{-\frac{\ell \eps}{6 d \ln2}} f(Y^{(t)})$, where we recall that $\ell=\min_{i\geq1}n_i$.
\item A lower bound on $f(Y^{(t)})$ for all $t$. This follows from properties A and B above.
\end{enumerate}
The above three steps imply that we must achieve $\ds(Y^{(t)})<\eps$ for some $t \in [T]$.

Step $1$ is immediate and step $2$ follows from the $G$-invariance of the capacity (property B) and a quantative form the AM-GM inequality (\cref{cor:lsw}).

Step $3$, or rather the lower bound on $\capac(X)$~(property A) is the technically most challenging part of the proof.
It is achieved by quantifying the basic argument given on p.~\pageref{page:qualitative}.
The main tool required to carry out this analysis is the existence of a set of invariant polynomials with ``small" integer coefficients that spans the space of invariant polynomials.
We do so by giving an ``explicit" construction of such a set via Schur-Weyl duality (\cref{subsec:invariants}).
Remarkably, we do not need to use the known exponential bounds on the degree of generating invariants in~\cite{derksen2001polynomial}.

\medskip

We give an alternative lower bound on $\capac(X)$ using Cayley's Omega-process (\cref{sec:cayley}).
This proof is more general (it applies to arbitrary actions of product of $\SL$'s) but gives a weaker bound (although sufficient to yield a polynomial time analysis for the algorithm).
Here we do need to rely on recent exponential bounds on the degree of generating invariants by Derksen.

The size of the integer coefficients of a spanning set of invariant polynomials is an interesting invariant theoretic quantity that appears to not have been studied in the invariant theory literature.
It is crucial for our analysis, and we believe it could be of independent interest in computational invariant theory.

\subsection{Capacity bound from invariant polynomials}\label{subsec:invariants}

If a tensor~$X$ is not in the null cone then its $G$-orbit remains bounded away from zero, i.e., $\capac(X)>0$.
In this section we will derive a quantitative version of this statement (\cref{prp:inf norm on orbit}).
As mentioned in the preceding section, this bound features as an important ingredient to the analysis our scaling algorithm (see \cref{subsec:detailed analysis algo} below).
To prove the bound, we use the concrete description of invariants in a crucial way.

We start by deriving a concrete description of the space of $G$-invariant polynomials on $V=\Ten_{n_0,\dots,n_d}(\C)$ of fixed degree~$m$, denoted $\C[V]^G_m$; see, e.g., \cite{procesi2007lie,burgisser2013explicit} for similar analyses.
Any such polynomial can be written as $P(X) = p(X^{\ot m})$, where $p$ is a $G$-invariant linear form on $V^{\ot m}$.\footnote{In general, $p$ is not unique. It would be unique if we required $p$ to be an element in the symmetric subspace.}
We are thus interested in the vector space
\begin{align*}
\bigl( ( V^{\ot m} )^*\bigr)^G
= ( ( \C^{n_0} )^{\ot m} )^* \ot
   \bigl( ( ( \C^{n_1} )^{\ot m} )^* \bigr)^{\SL(n_1)} \ot\ldots\ot \bigl( ( ( \C^{n_{d}} )^{\ot m} )^* \bigr)^{\SL(n_{d})} .
\end{align*}
Here $V^*$ is the dual vector space of $V$, i.e., the space of linear functions $f: V \rightarrow \C$. The action of $G$ on $V$ induces an action of $G$ on $V^*$ as follows: $(g \cdot f)(v) = f(g^{-1} \cdot v)$. The notation $W^G$ denotes the subspace of $W$ invariant under
the specified action of the group $G$ on $W$. We first focus on a single tensor factor.
Clearly, the $n$ by $n$ determinant is an $\SL(n)$-invariant linear form in $( (\C^n)^{\ot n} )^*$.
More generally, if $m$ is a multiple of $n$, then for each permutation $\pi\in S_m$,
\begin{equation}\label{eq:multiple determinants}
p_{n,\pi}(v_1 \ot\dots\ot v_m) = \det(v_{\pi(1)},\dots,v_{\pi(n)}) \det(v_{\pi(n+1)},\dots,v_{\pi(2n)}) \cdots \det(v_{\pi(m-n+1)},\dots,v_{\pi(m)})
\end{equation}
defines an $\SL(n)$-invariant linear form in $((\C^n)^{\ot m})^*$.
In fact:

\begin{lemma}\label{lem:single tensor factor}
If $n$ divides $m$ then the linear forms $p_{n,\pi}$ for $\pi\in S_m$ span the space of $\SL(n)$-invariant forms in $( ( \C^{n} )^{\ot m} )^*$.
If $n$ does not divide $m$ then no nonzero $\SL(n)$-invariant linear forms exist.
\end{lemma}
\begin{proof}
The space $((\C^n)^{\ot m})^*$ is not only a representation of $\SL(n)$, which acts by $m$-fold tensor powers, but also of the permutation group~$S_m$, which acts by permuting the $m$ tensor factors.
Both actions clearly commute, so the subspace of $\SL(n)$-invariant forms is again a representation of~$S_m$.
A fundamental result in representation theory known as Schur-Weyl duality implies that the space of $\SL(n)$-invariant forms is in fact an \emph{irreducible} representation of $S_m$ if $n$ divides $m$ (and that otherwise it is zero).
This implies that the space of $\SL(n)$-invariant forms is spanned by the $S_m$-orbit of any single nonzero $\SL(n)$-invariant form.
It is clear that the linear forms $p_{n,\pi}$ form a single $S_m$-orbit; since $p_{n,\operatorname{id}}(e_1\ot\dots\ot e_n\ot\dots\ot e_1\ot\dots\ot e_n) \neq 0$ this establishes the lemma.
\end{proof}

Lastly, we denote by $\varepsilon_{i_1,\dots,i_m}$ the dual basis of the standard product basis of $(\C^{n_0})^{\ot m}$, i.e., $\varepsilon_{i_1,\dots,i_m}(e_{j_1}\ot\dots\ot e_{j_m}) = \delta_{i_1,j_1}\cdots\delta_{i_m,j_m}$.
As a direct consequence of \cref{lem:single tensor factor} and our preceding discussion, we find that the polynomials
\begin{align}\label{eq:spanning invariants}
P(X) = (\eps_{i_1,\dots,i_m} \ot p_{n_1,\pi_1}\ot\dots\ot p_{n_{d},\pi_{d}})(X^{\ot m}) ,
\end{align}
where $\pi_1,\dots,\pi_{d}\in S_m$ and $i_1,\dots,i_m\in[n_0]$,
span the space of $G$-invariant polynomials of degree $m$ on $V$ if $n_1|m,\dots,n_{d}|m$
(otherwise there exist no $G$-invariant polynomials).

This concrete description of $\C[V]^G_m$ allows us to establish a fundamental bound on invariants:

\begin{proposition}\label{prp:bound on invariants}
The vector space $\C[V]^G_m$ of $G$-invariant homogeneous polynomials is nonzero only if $n_1|m,\dots,n_{d}|m$.
In this case, it is spanned by the polynomials $P(X)$ defined in~\eqref{eq:spanning invariants}, where $\pi_1,\dots,\pi_{d}\in S_m$ and $i_1,\dots,i_m\in[n_0]$.
These are polynomials with integer coefficients and they satisfy the bound
\[ \abs{P(X)} \leq (n_1\cdots n_{d})^m \, \norm{X}^m_2. \]
\end{proposition}

\begin{proof}
The determinant is a homogeneous polynomial with integer coefficients and it is clear that this property is inherited by the polynomials $P(X)$, so it remains to verify the bound.
Let $X\in V=\Ten_{n_0,\dots,n_d}(\C)$ be an arbitrary tensor.
We expand
$X = \sum_{j_1=1}^{n_1}\dots\sum_{j_d=1}^{n_d} v_{j_1,\dots,j_d} \ot e_{j_1} \ot \dots \ot e_{j_d}$,
where the $v_{j_1,\dots,j_d}$ are vectors in $\C^{n_0}$.
Note that the euclidean norm of
$\eps_{i_1,\dots,i_m} (\bigotimes_{\alpha=1}^m v_{j_1,\ldots,j_d})$
is bounded by $\|X\|_2^m$.
We obtain
\begin{align*}
X^{\ot m}
= \sum_{J_1\colon[m]\to[n_1]}\dots\sum_{J_d\colon[m]\to[n_d]} \left( \bigotimes_{\alpha=1}^m v_{J_1(\alpha),\dots,J_d(\alpha)} \right) \ot \left( \bigotimes_{\alpha=1}^m e_{J_1(\alpha)} \right) \ot\dots\ot \left( \bigotimes_{\alpha=1}^m e_{J_d(\alpha)} \right) ,
\end{align*}
and so
\begin{align*}
P(X)
&= (\eps_{i_1,\dots,i_m} \ot p_{n_1,\pi_1}\ot\dots\ot p_{n_{d},\pi_{d}})(X^{\ot m}) \\
&= \sum_{J_1,\dots,J_d} \eps_{i_1,\dots,i_m}\left(\bigotimes_{\alpha=1}^m v_{J_1(\alpha),\dots,J_d(\alpha)} \right) \;
p_{n_1,\pi_1}\left(\bigotimes_{\alpha=1}^m e_{J_1(\alpha)}\right) \cdots p_{n_{d},\pi_{d}}\left(\bigotimes_{\alpha=1}^m e_{J_{d}(\alpha)}\right) .
\end{align*}
The first factor is bounded in absolute value by $\|X\|_2^m$.
In view of \cref{eq:multiple determinants}, the remaining factors are determinants
of matrices whose columns are standard basis vectors; hence they are equal to zero or $\pm1$.
Altogether, we find that
\begin{equation*}
\lvert P(X) \rvert
\leq \sum_{J_1,\dots,J_{d}} \norm{X}_2^m = (n_1 \cdots n_{d})^m \norm{X}^m_2 . \qedhere
\end{equation*}
\end{proof}

We remark that similar results can be established for covariants.
We now use \cref{prp:bound on invariants} to deduce the following lower bound on the capacity:

\begin{proposition}\label{prp:inf norm on orbit}
Let $X\in\Ten_{n_0,\ldots, n_{d}}(\Z)$ be a tensor with integral entries that is not in the null cone.
Then,
\[ \capac(X) = \inf_{g\in G} \lVert g \cdot X \rVert^2_2 \geq \frac 1 {(n_1 \cdots n_d)^2} \geq \frac 1 {n^2}, \]
where we recall that $n = n_0 n_1 \cdots n_d$.
\end{proposition}
\begin{proof}
Since $X$ is not in the null cone, there exists some $G$-invariant homogeneous polynomial~$P$ as in \cref{prp:bound on invariants}
such that $P(X)\neq0$.
Let $m$ denote the degree of $P$.
Then it holds that
\[ 1 \leq \abs{P(X)} = \abs{P(g \cdot X)} \leq (n_1\cdots n_{d})^m \lVert g \cdot X \rVert^m_2. \]
The first inequality holds because both $P$ and $X$ have integer coefficients, the subsequent equality because $P$ is a $G$-invariant polynomial, and the last inequality is the bound from \cref{prp:bound on invariants}.
From this the claim follows immediately.
\end{proof}

In \cref{sec:cayley} we explain how rather worse bounds, that are nevertheless sufficient to establish \cref{thm:main}, can be established by analyzing the Cayley Omega-process~\cite{Derksen-Kemper}, a classical tool from invariant theory, and Derksen's degree bounds~\cite{derksen2001polynomial}.

\subsection{Analysis of the algorithm}\label{subsec:detailed analysis algo}

In this section, we formally analyze \cref{alg:scaling} and prove \cref{thm:main precise}.
As mentioned in \cref{subsec:scaling analysis}, key to our analysis is the bound on the norm of vectors in a $G$-orbit established in \cref{prp:inf norm on orbit}, as well the following result from~\cite{LSW} (see~\cite[Lemma 5.1]{GGOW}), which will allow us to make progress in each step of the scaling algorithm.


\begin{lemma}\label{cor:lsw}
  Let $x_1, \ldots, x_n$ be positive real numbers such that $\sum_{i=1}^n x_i = n$
  and $\sum_{i=1}^n (x_i-1)^2 \geq \beta$, where $\beta\leq1$. Then,
  $\prod_{i=1}^n x_i \leq 2^{-\beta/(6\ln2)}$.
\end{lemma}

We will also need the following easy lemmas which are well-known in the quantum information literature.
We supply proofs for sake of completeness:

\begin{lemma}\label{lem:singular marginal implies null cone}
Let $X\in\Ten(n_0,\dots,n_d)$ be a nonzero tensor with a singular
reduced density matrix $\rho_i$ for some $i\ge 1$.
Then $X$ is in the null cone for the $G$-action.
\end{lemma}
\begin{proof}
Let $P_i$ denote the orthogonal projection onto the support of $\rho_i$ and $Q_i = I_{n_i} - P_i$.
Note that 
$Q_i\neq0$ since $\rho_i$ is singular.
Then the following block diagonal matrix is in $\SL(n_i)$,
\[ g_i := \eps^{1/\rk(P_i)} P_i + \eps^{-1/\rk(Q_i)} Q_i, \]
and it is clear that
\[ g_i \cdot X = \eps^{1/\rk(P_i)} X \to 0 \]
as $\eps\to0$.
Thus $X$ is in the null cone.
\end{proof}

\begin{lemma}[\cite{dur2000three}]\label{lem:nonsingular marginals}
The rank of the quantum marginals $\rho_i$ associated with a tensor $X\in\Ten(n_0,\dots,n_d)$ is invariant under scaling.
\end{lemma}
\begin{proof}
Let $X\in\Ten(n_0,\dots,n_d)$ and $\rho_i$ be one of its quantum marginals.
Then the rank of $\rho_i$ is equal to the rank of $X$ considered as a 2-tensor
(namely, a matrix) in $\C^{n_i} \ot (\bigotimes_{j=0,j\neq i}^d \C^{n_j})$.
But the latter is well-known to be invariant under $\GL(n_i) \times \GL(\prod_{j=0,j\neq i}^d n_j)$.
In particular, it is invariant when we apply $A_i \ot (I_{n_0} \ot \otimes_{j=1,j\neq i}^d A_j)$ for arbitrary invertible matrices $A_1,\dots,A_d$.
Thus the rank of $\rho_i$ is a scaling invariant.
\end{proof}

We now prove \cref{thm:main precise}, restated for convenience, following the outline given in \cref{subsec:scaling analysis}.

\MainTheoremPrecise*
\begin{proof}
If we return in step~\ref{it:scaling step 1}, then $X$ is correctly identified to be in the null cone by \cref{lem:singular marginal implies null cone}.
Now assume that the quantum marginals are all nonsingular so that the algorithm proceeds to step~\ref{it:scaling step 2}.
Let us denote by $i^{(t)}$ the indices selected at the algorithm in steps $t=1,2,\dots$ and by $\rho^{(t)}_{i^{(t)}}$ the corresponding marginals.

We first note that the inverse matrix square roots are always well-defined and that each $Y^{(t)}$ is a scaling of $X$.
Indeed, this is clear for $t=1$, and since $Y^{(t+1)} = A^{(t)}_{i^{(t)}} \cdot Y^{(t)}$, where
\[ A^{(t)}_{i^{(t)}} := \det(\rho^{(t)}_{i^{(t)}})^{1/2n_{i^{(t)}}} (\rho^{(t)}_{i^{(t)}})^{-1/2} \in \SL(n_{i^{(t)}}), \]
we have that $Y^{(t+1)}$ is a scaling $Y^{(t)}$, which in turn implies that its quantum marginals are nonsingular (\cref{lem:nonsingular marginals}).
In particular, it follows that if \cref{alg:scaling} returns in step~\ref{it:scaling step 2} then it returns a scaling $Y$ such that $\ds(Y)<\eps$.

To finish the proof of the theorem, let us assume \cref{alg:scaling} does \emph{not} return in step~\ref{it:scaling step 2}
but instead proceeds to step~\ref{it:scaling step 3}.
We follow the sketch in \cref{subsec:scaling analysis} and track the norm squared of the tensors $Y^{(t)}$.
Initially,
\begin{equation}\label{eq:initial bound}
  \norm{Y^{(1)}}_2^2 = \Bigl\lVert X \Bigr\rVert_2^2 \leq n \, 2^{2b}
\end{equation}
since $b$ is a bound on the bitsize of each entry in the tensor.
Next, we note in each step of the algorithm,
\begin{align*}
&\quad \norm{Y^{(t+1)}}_2^2
= \det(\rho^{(t)}_{i^{(t)}})^{1/n_{i^{(t)}}} \norm{(\rho^{(t)}_{i^{(t)}})^{-1/2} \cdot Y^{(t)}}_2^2
= \det(n_{i^{(t)}} \rho^{(t)}_{i^{(t)}})^{1/n_{i^{(t)}}} \norm{(n_{i^{(t)}} \rho^{(t)}_{i^{(t)}})^{-1/2} \cdot Y^{(t)}}_2^2 \\
&< 2^{-\frac{n_i \eps}{6 d \ln2}} \norm{(n_{i^{(t)}} \rho^{(t)}_{i^{(t)}})^{-1/2} \cdot Y^{(t)}}_2^2,
\end{align*}
where we used \cref{cor:lsw} (with $\beta=n_i^2 \eps/d \leq 1$) and the choice of the subsystem $i^{(t)}$ in our algorithm.
We now show that the right-hand norm squared is precisely the norm squared of $Y^{(t)}$:
\begin{align*}
&\quad \norm{(n_{i^{(t)}} \rho^{(t)}_{i^{(t)}})^{-1/2} \cdot Y^{(t)}}_2^2
= \tr[ \left(I_{n_0}\ot\dots\ot I_{n_{i^{(t)}-1}} \ot (n_{i^{(t)}} \rho^{(t)}_{i^{(t)}})^{-1} \ot I_{n_{i^{(t)}+1}}\dots \ot I_{n_d}\right) Y^{(t)} (Y^{(t)})^\dagger ] \\
&= \tr[ \left(I_{n_0}\ot\dots\ot I_{n_{i^{(t)}-1}} \ot (n_{i^{(t)}} \rho^{(t)}_{i^{(t)}})^{-1} \ot I_{n_{i^{(t)}+1}}\dots \ot I_{n_d} \right) \frac {Y^{(t)} (Y^{(t)})^\dagger} {(Y^{(t)})^\dagger Y^{(t)}} ] \norm{Y^{(t)}}_2^2 \\
&= \tr[ (n_{i^{(t)}} \rho^{(t)}_{i^{(t)}})^{-1} \rho^{(t)}_{i^{(t)}} ] \, \norm{Y^{(t)}}_2^2
= \tr[ \frac {I_{n_{i^{(t)}}}} {n_{i^{(t)}}} ] \norm{Y^{(t)}}_2^2
= \norm{Y^{(t)}}_2^2.
\end{align*}
Together, we find the crucial estimate
\begin{align*}
  \norm{Y^{(t+1)}}_2^2 < 2^{-\frac{n_i \eps}{6 d \ln2}} \norm{Y^{(t)}}_2^2
\end{align*}
and so, using \cref{eq:initial bound},
\begin{equation}\label{eq:crucial upper}
  \norm{Y^{(T+1)}}_2^2 < 2^{-\frac{T \ell \eps}{6 d \ln2}} \norm{X}_2^2 \leq 2^{-\frac{T \ell \eps}{6 d \ln2}} \cdot n 2^{2b},
\end{equation}
where we recall that $\ell = \min_{i\geq1} n_i$.
Now we assume that $X$ is not in the null cone.
We will show that this leads to a contradiction.
Since $X$ has integral entries we can use the lower bound in \cref{prp:inf norm on orbit}.
Using that $Y^{(T+1)}$ is a scaling of $X$ we obtain
\begin{equation}\label{eq:crucial lower}
\norm{Y^{(T+1)}}_2^2 \geq \capac\left(Y^{(T+1)}\right) = \capac(X) \geq \frac 1 {n^2}.
\end{equation}
Comparing \cref{eq:crucial upper,eq:crucial lower}, it follows that
\begin{align*}
T < \frac {18\ln2} {\ell \eps} d \left(  b + \log n \right),
\end{align*}
which is the desired contradiction.
Thus $X$ is in the null cone, which agrees with the output of the algorithm in step~\ref{it:scaling step 3}.
\end{proof}

\paragraph{Analysis of bit complexity.}
To get an actual algorithm, one needs to control the bit complexity of the intermediate computations in \cref{alg:scaling}.
We provide a rough sketch of how to do this truncation argument, following the notation of the proof of \cref{thm:main precise}, and leave the details to the reader (cf.~\cite{GGOW}).
What we will do is that, at each step, instead of multiplying $Y^{(t)}$ by $A^{(t)}_{i^{(t)}}$ to get $Y^{(t+1)}$, we will multiply $Y^{(t)}$ by $B^{(t)}$ to get $Y^{(t+1)}$.
Here $B^{(t)}$ is obtained by truncating the entries of $A^{(t)}_{i^{(t)}}$ to a fixed number $P(n,b,\log(1/\eps))$ of bits after the decimal point, where $P$ is an appropriately chosen polynomial.
This way we can ensure that $\|A^{(t)}_{i^{(t)}} - B^{(t)}\|_F^2 \le \delta \ll \eps$.
Note that the only data we maintain during the algorithm is $Y^{(t)}$.
The above truncation ensures that as long as we run for polynomially many iterations, the bit sizes of entries of $Y^{(t)}$ will remain polynomially bounded.
It is also crucial to observe the norm of $Y^{(t)}$ keeps decreasing, so the number of bits before the decimal point also remain bounded.
Since the action of $A^{(t)}_{i^{(t)}}$ decreases the norm-squared of $Y^{(t)}$ by a fixed factor, the action of $B^{(t)}$ also does the same.
One last detail is that since $B^{(t)}$ may not have determinant one after the truncation, we may longer remain in the $G$-orbit of the initial tensor, so how do we get a lower bound on $\norm{Y^{(T+1)}}_2^2$? This is not a problem since $\det(B^{(t)}) \ge (1-\sqrt{\delta})^{n_{i^{(t)}}}$ and hence we are ``almost" in the $G$-orbit if $T$ is small.
In particular, one can prove that $\norm{Y^{(T+1)}}_2^2 \ge \frac{1}{n} (1-\sqrt{\delta})^{2T}$ and this suffices for the analysis.

\section{Hilbert-Mumford criterion and quantifying instability}\label{sec:instability}

In this section, we explain a special case of the Hilbert-Mumford criterion, which allows us to characterize tensors that lie in the null cone in terms of simple, one-dimensional subgroups of our group $\SL(n_1) \times \cdots \times \SL(n_{d})$.
Moreover, we define the instability (\cref{def:ins}), which quantifies how fast a tensor in the null cone can be sent to the zero vector $0\in V$ by the group action.
The instability is an important quantity in invariant theory and we present an algorithm for the $(0,\eps)$-gap problem (\cref{thm:eps instability}).

We again consider briefly the general setting of a group~$G$ acting on a vector space~$V$.
We know from \cref{sec:nullcone as optimization} that a vector is in the null cone if and only if its orbit closure contains the zero vector, $0 \in \overline{G \cdot v}$.
The \emph{Hilbert-Mumford criterion}~\cite{Hil, Mum65} shows that it suffices to consider certain one-dimensional (and hence in particular commutative) subgroups of~$G$.
More precisely, recall that a \emph{one-parameter subgroup} (\emph{1-PSG}) of $G$ is an algebraic group homomorphism $\lambda: \C^* \rightarrow G$, i.e., an algebraic map that satisfies $\lambda(zw) = \lambda(z) \lambda(w)$.
(For example, any 1-PSG of $\C^*$ is of the form $z\mapsto z^k$ for some integer~$k$.)
Then the Hilbert-Mumford criterion states that a vector $v$ is in the null cone if and only if there exists a 1-PSG $\lambda$ of $G$ such that $\lim_{z\to0}\lambda(z)\cdot v=0$.

A familiar example appears when the group $G = \SL(n)$ acts by
left multiplication on a single matrix $X$.
In this case, the null cone consists of the singular matrices and the Hilbert-Mumford
criterion tells us that if $X$ is singular, then there is a 1-PSG given by
$\lambda(z) = B^{-1} \diag(z^{a_1}, \ldots, z^{a_n}) B$ which drives $X$ to zero.
It is easy to see that in this case such one-parameter subgroup
can be found by taking $B$ to be a matrix in $\SL(n)$ which makes $X$ upper triangular (through row eliminations and permutations) with its last
row being all zeros, and that the exponents $a_i$ can be taken such that
$a_1 = \cdots = a_{n-1}=1$ and $a_n = 1-n$.

We now return to our setup, where $G = \SL(n_1) \times \cdots \times \SL(n_{d})$.
Here any 1-PSG $\lambda$ is of the following form:
\begin{equation}\label{eq:1-PSG}
  \lambda(z) = \left(B_1^{-1} \diag(z^{a_{1,1}},\ldots, z^{a_{1,n_1}}) B_1, \ldots,  B_d^{-1} \diag(z^{a_{d,1}},\ldots, z^{a_{d,n_{d}}}) B_d\right),
\end{equation}
where $B_1, \ldots, B_{d}$ are invertible matrices ($B_i$ of dimension $n_i \times n_i$) and $(a_{i,j})$ is a tuple of integers in
\begin{equation}\label{eq:abelian 1PSG}
  \Gamma := \bigl\{ (a_{i,j})_{i\in[d],j\in[n_i]} \;\big|\; a_{i,j} \in \Z, \, \sum_{j=1}^{n_i} a_{i,j}=0 \text{ for } i\in[d] \bigr\}.
\end{equation}
Intuitively, the matrices $B_i$ are a change of basis after which the action of $G$ becomes reduced to an action of diagonal groups (similar to actions arising in the study of the classical matrix or tensor scalings).

We want to understand what it means for a 1-PSG $\lambda$ and a tensor $X$ that $\lim_{z\to0} \lambda(z) \cdot X=0$.
For this, we write the tensor in the basis corresponding to $B=(B_1,\dots,B_{d})$, i.e., $Y = B \cdot X$.
We define the \emph{support} of $Y$ as
$$
\supp(Y) := \left\{ (j_1,\ldots, j_{d}) \in [n_1] \times \cdots \times [n_{d}] \: \big| \: \exists j_0\: \text{s.t.}\: Y_{j_0,j_1,\ldots,j_d} \neq 0 \right\}.
$$
Thus $(j_1,\dots,j_{d})$ is in the support of $Y$ iff at least one of the slices $Y^{(i)}$ in the $0^\text{th}$ direction of the tensor~$Y$
has a nonzero entry at this position.
Now note that
\begin{equation}\label{eq:action of 1-PSG}
  (B \cdot \lambda(z)\cdot X)_{j_0,j_1\dots,j_d} = z^{\sum_{i=1}^{d} a_{i,j_i}} Y_{j_0,j_1,\dots,j_d} .
\end{equation}
It follows that $\lim_{z\to0} \lambda(z)\cdot X=0$ is equivalent to the support $\supp(Y)$ being deficient in the following sense:


\begin{definition}[Deficiency]
We call a subset $S \subseteq [n_1] \times \cdots \times [n_{d}]$ \emph{deficient} if there exists $(a_{i,j}) \in \Gamma$ such that
\[
 \forall \: (j_1,\ldots, j_{d}) \in S\quad \sum_{i=1}^{d} a_{i, j_i} > 0 .
\]
\end{definition}
We note that in the case $d=2$, a subset $S \subseteq [n] \times [n]$ is deficient if and only if the bipartite graph corresponding to $S$ does not admit a perfect matching, as can be proved via Hall's matching theorem.
For $d > 2$, we do not know of such a clean combinatorial characterization,
although the characterization above is given by a linear program, and therefore is efficiently testable.

In the Hilbert-Mumford criterion, one can do slightly better, namely restrict attention to the one-parameter subgroups compatible with a maximally compact subgroup of the group $G$.
What this means in our case is that we will be able to take the matrices $B_1,\ldots, B_{d}$ to be unitary matrices (see \cref{lem:instability S vs U})!
We can thus summarize the statement of the Hilbert-Mumford criterion as follows:

\begin{restatable}[Special case of~\cite{Hil, Mum65}]{proposition}{HilbertMumford}\label{prp:hilbert mumford}
A tensor $X \in \Ten(n_0,\ldots, n_d)$ is in the null cone of the natural action of the group $G = \SL(n_1) \times \cdots \times \SL(n_{d})$ iff there exist unitary $n_i \times n_i$ matrices $U_i$, $i\in[d]$, such that the support $\supp(Y)$ of the tensor $Y = (U_1,\ldots, U_{d}) \cdot X$ is deficient.
\end{restatable}

\noindent For completeness, we provide a proof of the criterion in~\cref{app:git} (following~\cite{sury00} by reducing it to Farkas' lemma of linear programming).

\begin{remark}
Deficiency can also be understood in terms of the null cone of the action of the subgroup
$T \subseteq \SL(n_1) \times \cdots \times \SL(n_{d})$ formed by tuples of diagonal matrices.
Indeed, when we fixed $B$ but varied the $a_{i,j}$, we were precisely varying over the 1-PSGs of $T$.
Thus, a tensor $Y \in \Ten(n_1,\dots,n_d)$ is in the null cone for the action of the diagonal subgroup if and only if $\supp(Y)$ is deficient.
(This follows from the Hilbert-Mumford criterion for the commutative group $T$, or directly from Farkas' lemma.)
Thus \cref{prp:hilbert mumford} is a special case of a general reduction principle in geometric invariant theory, from actions of non-commutative groups to commutative groups up to a basis change.
\end{remark}

In geometric invariant theory~\cite{Mum65}, vectors in the null cone are also referred to as \emph{unstable}.
The Hilbert-Mumford criterion suggests a natural way of quantifying the instability of a vector: instead of merely asking whether there exists a 1-PSG that sends the vector to zero, we may measure the (suitably normalized) rate at which the vector is sent to zero in \cref{eq:action of 1-PSG}.
This rate, also known as \emph{Mumford's numerical function}, takes the following form for the tensor action that we are studying:

\begin{definition}[Deficiency, instability]\label{def:ins}
Given a set $S\subseteq[n_1]\times\dots\times[n_d]$, we define its \emph{deficiency} as
\begin{align*}
  \defi(S) & := \max_{(a_{i,j}) \in \Gamma} \frac {\min_{(j_1,\dots,j_{d}) \in S}\left(\sum_{i=1}^{d} a_{i,j_i}\right)} {\sqrt{\sum_{i=1}^{d} \sum_{j=1}^{n_i} a_{i,j}^2}}.
\end{align*}
where we recall that the set $\Gamma$ was defined in \cref{eq:abelian 1PSG}.
We then define the \emph{instability} of a tensor~$X\in\Ten(n_0,\dots,n_d)$ by
\[ \ins(X) := \max_{\substack{U=(U_1,\dots,U_{d}) \text{ tuple}\\ \text{of unitary matrices}}} \defi(\supp(U \cdot X)), \]
\end{definition}

\noindent The instability can also be defined using general invertible matrices:

\begin{restatable}{lemma}{insSvsU}\label{lem:instability S vs U}
For any $X\in\Ten(n_0,\dots,n_d)$, we have that
\[ \ins(X) = \max_{\substack{B=(B_1,\dots,B_{d}) \text{ tuple}\\ \text{of invertible matrices}}} \defi(\supp(B \cdot X)) . \]
In particular, $X\mapsto \ins(X)$ is a $G$-invariant function.
\end{restatable}

\noindent In our case, this follows from the QR decomposition, and we give a succinct proof in \cref{app:git}.

Clearly, a subset~$S$ is deficient if and only if its deficiency is positive, $\defi(S)>0$.
And by the Hilbert-Mumford criterion, a tensor~$X$ is in the null cone if and only if its instability is positive, $\ins(X)>0$.
The latter suggests a natural relaxation of the null-cone problem:

\begin{problem}\label{prb:eps-instability}
  For $\eps>0$, the problem \emph{$\eps$-instability} is defined as follows:
  Given a tensor $X\in\Ten(n_0,\dots,n_d)$ with integer coordinates such that either
  \begin{enumerate}
    \item $X$ is not in the null cone $($i.e., $\ins(X)\leq0)$, or
    \item $\ins(X)\geq\eps$,
  \end{enumerate}
  decide which is the case.
\end{problem}

We will now show that the $\eps$-instability problem can be solved as a consequence of our main \cref{thm:main}.
Importantly, the instability of a tensor is tightly connected to the distance from $d$-stochasticity, as defined in \cref{def:ds}.

\begin{restatable}[Special case of {\cite[Lemma 3.1, (iv)]{NM84}}]{lemma}{insvsds}\label{lem:ins vs ds}
For all tensors $X\in\Ten(n_0,\dots,n_d)$,
\[ \ins(X) \leq \sqrt{\dds(X)}. \]
\end{restatable}

We prove \cref{lem:ins vs ds} in \cref{app:git}.
In fact, the inequality in \cref{lem:ins vs ds} is tight for tensors in the null cone (e.g., \cite[Corollary 11.2]{georgoulas2013moment}), but we will not need this here (if $X$ is not in the null cone then $\ins(X)\leq0$ while $\ds(X)=0$).
We obtain the following corollary of \cref{thm:main}.

\begin{theorem}\label{thm:eps instability}
There is $\poly\left(n, b, \frac1\eps \right)$ time deterministic algorithm (\cref{alg:scaling})
for the $\eps$-instability problem (\cref{prb:eps-instability}) for tensors with entries
of bit size bounded by~$b$.
\end{theorem}

The algorithm is obtained by running \cref{alg:scaling} and outputting ``$X$ is not in the null cone" if \cref{alg:scaling} produces a scaling $Y$ with $\ds(Y) < \eps$ and otherwise outputting ``$\ins(X)\geq\eps$''.

Thus $\eps$-instability can be decided efficiently if $\eps=\Omega(1/\poly(n))$.
Unfortunately, there may well exist tensors~$X$ which are in the null cone and whose instability is exponentially small.
However, the instability cannot be worse than exponentially small, as we prove in \cref{app:git}:

\begin{restatable}{lemma}{inslb}\label{lem:gap}
Suppose a tensor $X \in \Ten(n_0,\ldots,n_d)$ is in the null cone.
Then $\sqrt{\dds(X)} \geq \ins(X) = \exp(-O(n\log n))$,
where we recall that $n=n_0 \cdots n_d$.
\end{restatable}

This begs the question: are there natural examples where we can
instantiate \cref{alg:scaling} 
with $1/\epsilon = \poly(b,n)$?
One set of examples comes from the notion of slice rank of tensors -- discovered and widely used in the recent breakthroughs in additive combinatorics regarding cap sets and other combinatorial objects (see, e.g., \cite{BCCGNSU17}).
We discuss this in \cref{sec:slice rank}.

\section{Quantum marginals and Kronecker coefficients}\label{sec:quantum}

In this section, we introduce two versions of the quantum marginal problem
(\cref{prb:qmp} and \cref{prb:qmp slocc}), and elaborate on their relevance in
quantum information theory. Additionally, we describe the connection between the
quantum marginal problem and an important problem from representation theory:
the positivity of Kronecker coefficients.

As mentioned before, the problems we are studying have a very natural interpretation
from the perspective of quantum physics.
It will be convenient to set $n_0=1$; 
thus we study the action of $G=\SL(n_1)\times\dots\times\SL(n_d)$ on $\Ten(n_1,\dots,n_d)$.

Recall that the classical marginal problem for probability distributions ask the following question: given probability distributions of various subsets of a number of random variables, are they consistent with a joint probability distribution on all the random variables?
For example, if $p(x,y)$ and $p(y,z)$ are bivariate probability distributions then a necessary and sufficient condition for their compatibility is that $\sum_x p(x,y) = \sum_z p(y,z)$ (indeed, in this case we can define $p(x,y,z) = p(x,y) p(y,z) / \sum_x p(x,y)$).

The \emph{quantum marginal problem}~\cite{klyachko2004quantum} is the quantum analog of the classical marginal problem.
It asks whether given quantum states on various subsets of a number of quantum systems are consistent with a joint quantum state on all the systems.
In the context of this paper, we are particularly interested in the case when the quantum marginals $\rho_k$ of single subsystems are given and the overall state is to be in a pure state.
More formally, the problem is defined as follows:

\begin{problem}[One-body quantum marginal problem]\label{prb:qmp}
Given PSD matrices $\rho_1,\dots,\rho_d$ with unit trace (``density matrices''), does there exist a tensor $Y\in\Ten(n_1,\dots,n_d)$ (``pure state'') such that the $\rho_k$ are the quantum marginals of $Y Y^\dagger/Y^\dagger Y$ in the sense of \cref{def:quantum marginal}?
\end{problem}

This problem does not have a direct classical analog, since the probability theory analog of a pure state is a deterministic probability distribution, for which the problem is trivial.

A generalization of the one-body quantum marginal problem asks the same question but with the pure quantum state $Y$ varying over the $G$-orbit closure of a given tensor $X$ (see, e.g., \cite{walter2013entanglement} and the references therein).
The physical interpretation of this group action is that $Y$ can be obtained to arbitrary precision from $X$ by a class of quantum operations known as \emph{stochastic local operations and classical communication (SLOCC)}~\cite{bennett2000exact}.
SLOCC can be intuitively understood as follows: we imagine that different parties hold the different systems of a quantum state; SLOCC then corresponds to a sequence of local quantum operations and measurements, where we allow for \emph{post-selection} on specific measurement outcomes.
Thus we are interested in the following problem:

\begin{problem}[One-body quantum marginal problem, SLOCC version]\label{prb:qmp slocc}
Given PSD matrices $\rho_1,\dots,\rho_d$ with unit trace and a tensor $X\in\Ten(n_1,\dots,n_d)$, does there exist $Y \in \overline{G \cdot X}$ such that the $\rho_k$ are the quantum marginals of $YY^\dagger$?
\end{problem}

Now for the special case when $\rho_k = I_{n_k}/n_k$, we get precisely the problem we are studying in this paper (cf.~\cref{thm:kempf ness}).
In quantum theory, this means that the $k^\text{th}$ system is maximally entangled with the other subsystems, as we already briefly mentioned before.
Thus a quantum state with these marginals is called \emph{locally maximally entangled}, and we can interpret the transformation $X \leadsto Y$ as an \emph{entanglement distillation}\footnote{It can also be understood as finding a normal form of the quantum state $X$~\cite{VDB03}.}.

Therefore, \cref{alg:scaling} gives a method for distilling locally maximally entanglement from a given pure quantum state (not in the null cone).
More precisely, \cref{thm:main,thm:null cone exp scaling} imply the following result:

\begin{theorem}[SLOCC entanglement distillation]
There is a $\poly(n, b, \frac1\eps)$ time deterministic algorithm that, given the classical description of a pure quantum state $X\in\Ten(n_1,\dots,n_d)$, with integer coordinates of bit size bounded by~$b$, either identifies that $X$ cannot be transformed by SLOCC into a locally maximally entangled state or outputs an SLOCC operation $g\in G$ such that $\ds(g \cdot X)<\eps$.
In particular, \cref{prb:qmp slocc} for $\rho_k=I_{n_k}/n_k$ can be decided in time $\poly(b)\exp(O(n \log n))$.
\end{theorem}

Here we have used that \cref{alg:scaling} is straightforwardly modified to record an SLOCC operation $g\in G$ such that $Y = g \cdot X$.
This operation can then be implemented in the laboratory to \emph{provably} distill an approximately locally maximally entangled state%
\footnote{It is interesting to contrast this with the fact it was only very recently understood for which dimensions $n_1,\dots,n_d$ locally maximally states exist at all~\cite{bryan2017existence}.}.
It is also possible to implement \cref{alg:scaling} directly as a quantum algorithm, given many copies of the unknown quantum state~$X$, without performing quantum state tomography of the global state (but only of the quantum marginals). 
This can be done by a straightforward adaption of the method described in~\cite{walter2013entanglement}.

We note that for general $\rho_1,\dots,\rho_d$, a natural alternating minimization algorithm along the lines of \cref{alg:scaling} suggests itself, and it is an interesting open question whether this algorithm can be analyzed using similar methods as developed in this paper.
We return to this question in \cref{sec:conclusion}.

\medskip

The answer to either version of the quantum marginal problem depends only on the eigenvalues of the matrices $\rho_k$, since we can always achieve an arbitrary eigenbasis by applying local unitaries $U_1\ot\dots\ot U_d$ to $Y$.
Let us denote by $s_k(Y)$ the vector of eigenvalues of the quantum marginals of $Y Y^\dagger/Y^\dagger Y$, ordered nonincreasingly.
Then the solution of \cref{prb:qmp slocc} is encoded by the following object:
\[ \Delta(X) := \{ (s_1(Y),\dots,s_d(Y)) \in \R^{n_1+\dots+n_d} \;|\; Y \in \overline{G \cdot X} \}. \]
A remarkable theorem by Kirwan asserts that $\Delta(X)$ is a convex polytope~\cite{Kirwan84} (``moment polytope'').
In our case, $\Delta(X)$ is known as the \emph{entanglement polytope} of $X$~\cite{walter2013entanglement}.
For generic tensors $X$, it is known~\cite{brion1987image} (cf.~\cite{BI16,walter2013entanglement})
that \cref{prb:qmp slocc} reduces to \cref{prb:qmp}.
That is,
\[  \Delta(n_1,\dots,n_d) = \{ (s_1(Y),\dots,s_d(Y)) \in \R^{n_1+\dots+n_d} \;|\; Y \in\Ten(n_1,\dots,n_d) \}. \]
is likewise a convex polytope and it encodes the answer to \cref{prb:qmp}.
It is known that deciding membership in $\Delta(n_1,\dots,n_d)$ is in NP $\cap$ coNP~\cite{buergisser2017membership}, and our paper has in part been motivated by the desire to find a polynomial-time algorithm for this problem.

%

\medskip

\Cref{thm:kempf ness} established an intriguing connection between the existence of nonzero) invariants
and quantum states with maximally mixed marginals.
Now a $G$-invariant polynomial is nothing but a trivial representation of the group $G$ in the polynomial ring $\C[V]$.
It holds more generally that there is a relation between irreducible representations in $\C[V]$ and the existence of quantum states with certain quantum marginals.
This relation, which can be proved using similar methods from geometric invariant theory as we used in this paper, allows for a most transparent proof of the polyhedrality of~$\Delta(X)$ mentioned above~\cite{NM84}, and a long line of works have drawn out this intriguing connection (see, e.g., \cite{klyachko2002coherent,klyachko2004quantum,daftuar2005quantum,christandl2006spectra,christandl2007nonzero,ressayre2010geometric,brandao2016mathematics,vergne2014inequalities}).
In particular, it implies a remarkable relation between the one-body quantum marginal problem and the Kronecker coefficients from representation theory, which we describe next.

Kronecker coefficients describe the interrelationship between irreducible representations of the symmetric group.
They are prevalent in algebraic combinatorics and play an important role in the GCT program.
The irreducible representations of the symmetric group $S_r$ are indexed by the partitions of $r$ (see, e.g., \cite{Sagan}).
Let us denote the irreducible representation corresponding to a partition $\lambda$ as $[\lambda]$.
The \emph{Kronecker coefficient} $k_{\lambda,\mu,\nu}$ is the multiplicity of the irreducible representation $[\lambda]$ in the representation $[\mu] \otimes [\nu]$ (it is symmetric in its three indices).
Determining the positivity of Kronecker coefficients, $k_{\lambda,\mu,\nu} > 0$, has long been a major challenge in representation theory and was recently proven to be NP-hard~\cite{ikenmeyer2017vanishing} (even if $\lambda, \mu, \nu$ are given in the unary encoding).
However, the asymptotic or \emph{stretched} version of this question is nicer and has a lot of structure.
It asks: given three partitions $\lambda, \mu, \nu$ of $r$, does there exist an integer $l \ge 1$ s.t.
$k_{l \lambda,l \mu, l \nu} > 0$.%
\footnote{The closely connected \emph{Littlewood-Richardson coefficients} $c_{\lambda, \mu, \nu}$ have the \emph{saturation property}, $c_{\lambda, \mu, \nu} > 0$ iff $c_{l \lambda, l \mu, l \nu} > 0$~\cite{KnT03}, which the Kronecker coefficients don't share.
Also, the positivity of Littlewood-Richardson coefficients can be decided in strongly polynomial time~\cite{deloera2006computation,GCT3,BI_LRcoeffs}.}
Quite astonishingly, this asymptotic version is equivalent to the one-body quantum marginal problem (for $3$ systems)!
That is, $k_{l \lambda,l \mu, l \nu} > 0$ for some $l \ge 1$ if and only if $(\lambda/r,\mu/r,\nu/r) \in \Delta(n,n,n)$, where $n$ is chosen to be at least as large as the number of parts in each of the three partitions.
This holds because the Kronecker coefficients are precisely the multiplicities of certain irreducible representations in the polynomial ring $\C[V]$.
Hence the asymptotic version is in NP $\cap$ coNP and believed to have a polynomial time algorithm.

This is yet another instance (compare the slice rank discussed in \cref{sec:slice rank} below),
where a problem is combinatorially nasty and NP-hard, but its asymptotic version is analytically nice and might well admit a polynomial time algorithm.
We view our alternating minimization algorithm and its analysis as a first step towards obtaining efficient algorithms for computing the asymptotic slice rank, for solving the one-body quantum marginal problems, and for deciding the asymptotic vanishing of Kronecker coefficients.
To achieve this, the main issue is an exponential improvement of speed
by obtaining a running time polynomially bounded
in $\log\frac{1}{\epsilon}$ instead of $\frac{1}{\epsilon}$, as in Theorem~\ref{thm:main precise}.

\section{Slice rank and the null cone}\label{sec:slice rank}

In this section, we define the slice rank of a tensor, a notion introduced by
Tao~\cite{Tao_slicerank} and developed further in a number of papers. We
describe the relation between an asymptotic version of the slice rank and the instability
of a tensor, therefore relating the slice rank of a tensor to the null-cone problem.
We will focus on tensors of format $n_0= 1$ and $n_1 = \cdots = n_{d} = m$.
We denote the space of such tensors by $\Ten_m^d := \Ten(m,\dots,m) \simeq\Ten(1,m,\dots,m)$.

\begin{definition}[Slice rank]
A tensor $X \in \Ten_m^d$ is said to have \emph{slice rank-one} if it decomposes as the tensor product of a vector and a lower-order tensor, i.e., if there exists a $k \in [d]$, a vector $V \in \C^m$, and a tensor $Y \in \Ten_m^{d-1}$ such that
\[
X_{j_1,\dots,j_d} = V_{j_k} \cdot Y_{j_1,\dots,j_{k-1},j_{k+1},\dots,j_d}.
\]
The \emph{slice rank} of a tensor $X$, denoted $\sr(X)$, is defined as the minimum number $r$ of slice rank-one tensors $X_1,\dots,X_r$ such that $X=\sum_{i=1}^r X_i$.
\end{definition}

Importantly, the coordinate $k$ along which we decompose each slice rank-one tensor $X_i$ may be different for each $i$.
Thus the slice rank is in general only upper-bounded by the rank of any flattening of the tensor.
In particular, it is clear that $\sr(X) \le m$ by flattening/slicing along any fixed coordinate.
The recent paper~\cite{BCCGNSU17} developed some intriguing connections between slice rank and the null cone of the action of the group $G = \SL(m)^d = \SL(m) \times \cdots \times \SL(m)$ ($d$ times):

\begin{restatable}[\cite{BCCGNSU17}]{theorem}{slicerank}\label{thm:slice rank1}
Consider a tensor $X \in \Ten_m^d$. If $\sr(X) < m$, then $X$ is in the null cone for $\SL(m)^d$. In fact, $\ins(X) \ge 1/\sqrt{d m^3}$.
\end{restatable}

\begin{theorem}[\cite{BCCGNSU17}]\label{thm:slice rank2}
Consider a tensor $X \in \Ten_m^d$. If $X$ is in the null cone for $\SL(m)^d$, then $\sr(X^{\otimes k}) < m^k$ for some positive integer $k$.
\end{theorem}

The notion of instability we consider is slightly different from the notion of instability considered in~\cite{BCCGNSU17} (since we use Mumford's original definition, there are differences in normalization; however, it is not hard to see that the two notions are polynomially related).
Hence we will provide a proof of \cref{thm:slice rank1} in~\cref{app:git}.
There we also prove the following lemma.

\begin{restatable}{lemma}{asymptstab}\label{le:asympt-stab}
Let $k \ge 1$. Then
$X$ is in the null cone of the action of $\SL(m)^d$ if and only if
$X^{\otimes k}$ is in the null cone of the action of $\SL(m^k)^d$.
\end{restatable}

\Cref{le:asympt-stab}, combined with the above theorems, implies that the asymptotic non-fullness of the slice rank characterizes the null cone:

\begin{corollary}\label{cor:slice rank null cone}
A tensor $X \in \Ten_m^d$ is in the null cone iff $\sr(X^{\otimes k}) < m^k$ for some positive integer $k$.
\end{corollary}

A tighter connection between the asymptotic slice rank and moment polytopes (images of moment maps discussed above) will be developed in~\cite{CGNOS17}.

We remark that the problem of determining the slice rank of a tensor $X \in \Ten_m^d$ is NP-hard for $d \ge 3$~\cite{Sawin17}.
It is plausible that the problem of determining whether $\sr(X) < m$ is NP-hard as well, but we do not know of a proof. 
However, the \emph{asymptotic} non-fullness of slice rank is equivalent to the null-cone problem, which we believe to be in polynomial time and make partial progress towards in this paper.
Indeed, using the scaling algorithm developed in this paper (\cref{thm:eps instability}), \cref{cor:slice rank null cone}, and the lower bound on the instability in \cref{thm:slice rank1}, we obtain the following corollary:

\begin{corollary}
There is a deterministic $\poly(m^d, b)$ time algorithm that, given a tensor $X \in \Ten_m^d$,
with entries of bit size bounded by $b$, such that either
\begin{enumerate}
\item $\sr(X^{\otimes k}) = m^k$ for all $k$, or
\item $\sr(X) < m$,
\end{enumerate}
decides which is the case.
\end{corollary}

We may also replace the second condition by $\sr(X^{\otimes k_0}) < m^{k_0}$ for some constant $k_0$.

\section{Cayley's Omega process and bounds on coefficients of generators of \texorpdfstring{$\C[V]^G$}{C[V]\^{}G}}\label{sec:cayley}

In this section, we prove upper bounds on the {\em coefficients} of the generating invariants for very general group actions on {\em any} vector space, namely those which can be decomposed into the actions of
special linear groups.
An important corollary of this result is that we obtain an alternative upper bound for the coefficients of the invariants of the
tensor product action of the group $G = \SL(n_1) \times \cdots \times \SL(n_d)$ on $V=\C^{n_0}\ot\C^{n_1}\ot\dots\ot\C^{n_d}$, studied in the previous sections.

As we prove this bound, we also devise a way of computing a set of generating invariants in exponential time.
This allows us to give an alternative (algebraic rather than analytic)
exponential algorithm (\cref{alg:null cone alternative}) for deciding membership in the null cone of tensor product action of $G = \SL(n_1) \times \cdots \times \SL(n_d)$ which generalizes immediately to arbitrary polynomial actions of~$G$.

\subsection{Reynolds operator}

We now introduce the Reynolds operator, which as we shall soon see, gives us an approach to computing a generating set of the invariant polynomials of a given degree.
The general definition of the Reynolds operator is as follows (cf.~\cite[\S2.2.1]{Derksen-Kemper}):

\begin{definition}[Reynolds operator]
For a reductive group~$G$ acting on a vector space~$V$, the Reynolds operator~$\RO$ is a $G$-invariant projection
\begin{align*}
\RO: \mathbb{C}[V] \rightarrow \mathbb{C}[V]^G.
\end{align*}
That is, it is a linear map satisfying the following properties for every polynomial $p \in \mathbb{C}[V]$:
\begin{itemize}
\item $\RO(p) = p$ if $p \in \mathbb{C}[V]^G$,

\item $\RO(g\cdot p) = \RO(p)$ for all $g \in G$.
\end{itemize}
\end{definition}

\noindent It is known that the Reynolds operator exists and that it is unique. It can be defined by averaging $g\cdot p$ over a maximal compact subgroup of~$G$ with respect to its invariant Haar measure.

The following is an immediate consequence of the defining properties of the Reynolds operator.

\begin{lemma}\label{span_inv} The set $\{\RO(\nu) \mid \nu \text{ monomial of degree } D\}$
generates the vector space $\mathbb{C}[V]_D^G$ of $G$-invariant polynomials of degree~$D$.
\end{lemma}

Thus, to understand the invariant ring $\C[V]^G$ and the null cone, we only need to understand the behaviour of the Reynolds operator on monomials.
Moreover, a theorem by Derksen~\cite{derksen2001polynomial} tells us that to determine a set of generators for $\C[V]^G$, one only needs to compute the invariants of degree at most exponential in the dimension of the group $G$.

With these observations in mind, given any group $G$ acting on a space $V$, an important algorithmic question is to efficiently compute the Reynolds operator $\RO$.
For this, a useful property of the Reynolds operator is that it can be computed from the Reynolds operators of any normal subgroup~$N$ of~$G$ and the corresponding quotient group~$G/N$.
This property is formalized in the following lemma:

\begin{lemma}[{\cite[Lemma 4.5.3]{Derksen-Kemper}}]\label{lem:reynolds-composition}
  If $G$ has a normal subgroup $N$ and $\cR_G, \cR_N, \cR_{G/N}$ are the Reynolds operators with respect to the groups $G$, $N$ and $G/N$, respectively, then we have $\cR_G = \cR_N \cR_{G/N}$.
\end{lemma}

This lemma essentially tells us that if $G$ can be decomposed into normal subgroups and quotient groups whose Reynolds operators are individually easy to compute, then $\RO$ is also easy to compute.
In \cref{subsec:coeff-bds} we will make use of this lemma to compute the Reynolds operator of our $d$-fold product action $G = \SL(n_1) \times \cdots \times \SL(n_d)$ via the Reynolds operators for the groups $\SL(n_i)$.
This will allow us to deduce good coefficient bounds on a set of generators for the invariants in $\C[V]^G$.

This is where Cayley's Omega process comes in.
The Omega process is used to construct the invariants for the groups $\SL(n_i)$, as we will now discuss.

\subsection{Cayley's Omega process}

For actions of the special linear group $\SL(m)$, Cayley's Omega process provides an
explicit way of constructing the Reynolds operator. We now define the Omega process
and study some of its properties.

Let $Z=(Z_{i,j})_{i,j=1}^m$ denote a matrix of variables and $q$ be a polynomial in $Z$.
The group $\GL(m)$ acts on the space of such polynomials:
for a matrix $A \in \GL(m)$, we define the polynomial $A \star q$ by $(A \star q)(Z) := q\left(A^{T} Z\right)$.
Clearly, this action preserves
homogeneity and degree.

\begin{definition}[Omega operator]\label{def:omega}
Given a polynomial $q$ in the variables
$(Z_{i,j})_{i,j=1}^m$, we define
\begin{align*}
\Omega(q) := \sum_{\sigma \in S_m}(-1)^{\sign(\sigma)} \frac{\partial^m q}{\partial Z_{1,\sigma(1)} \cdots \partial Z_{m,\sigma(m)}}.
\end{align*}
We also denote by $\Omega^r(q)$ the result of applying the differential operator~$\Omega$ $r$ times, starting with~$q$.
\end{definition}

If $q$ is homogeneous of degree $t$, then $\Omega^r(q)$ is homogeneous of degree $t-rm$.
In particular, $\Omega^r(q)$ is a constant if $q$ has degree $rm$.
The following is the crucial invariance property of the $\Omega$-process.

\begin{proposition} For any polynomial in $Z$ and
$A \in \GL(m)$ we have
$$
\Omega\left(A \star q\right) = \det(A) \cdot (A \star \Omega(q)) .
$$
\end{proposition}
\begin{proof}
By the chain rule,
\begin{align*}
  \frac{\partial \left( A \star q \right)}{\partial Z_{i,j}}
= \sum_{k=1}^m A_{i,k} \left( A \star \frac{\partial q}{\partial Z_{k,j}} \right).
\end{align*}
Iterating $m$ times, we get
\begin{align*}
  \frac{\partial^m (A \star q)}{\partial Z_{i_1,j_1} \cdots \partial Z_{i_m,j_m}}
= \sum_{k_1=1,\dots,k_m=1}^m \left( \prod_{t=1}^m A_{i_t,k_t} \right) \left( A \star \frac{\partial q}{\partial Z_{k_1,j_1} \cdots \partial Z_{k_m,j_m}} \right).
\end{align*}
Hence,
\begin{align*}
  \Omega(A \star q)
&= \sum_{\sigma \in S_m}(-1)^{\sign(\sigma)} \frac{\partial^m (A \star q)}{\partial Z_{1,\sigma(1)} \cdots \partial Z_{m,\sigma(m)}} \\
&= \sum_{\sigma \in S_m}(-1)^{\sign(\sigma)} \sum_{k_1=1,\dots,k_m=1}^m \left( \prod_{t=1}^m A_{t,k_t} \right) \left( A \star \frac{\partial q}{\partial Z_{k_1,\sigma(1)} \cdots \partial Z_{k_m,\sigma(m)}} \right) \\
&= \sum_{\sigma \in S_m}(-1)^{\sign(\sigma)} \sum_{\kappa \colon [m]\to[m]} \left( \prod_{t=1}^m A_{t,\kappa(t)} \right) \left( A \star \frac{\partial q}{\partial Z_{\kappa(1),\sigma(1)} \cdots \partial Z_{\kappa(m),\sigma(m)}} \right) \\
&= \sum_{\sigma \in S_m}(-1)^{\sign(\sigma)} \sum_{\kappa \colon [m]\to[m]} \left( \prod_{t=1}^m A_{t,\kappa(\sigma(t))} \right) \left( A \star \frac{\partial q}{\partial Z_{\kappa(1),1} \cdots \partial Z_{\kappa(m),m}} \right) \\
&= \sum_{\kappa \colon [m]\to[m]} \left( A \star \frac{\partial q}{\partial Z_{\kappa(1),1} \cdots \partial Z_{\kappa(m),m}} \right) \left( \sum_{\sigma \in S_m}(-1)^{\sign(\sigma)} \prod_{t=1}^m A_{t,\kappa(\sigma(t))} \right).
\end{align*}
It is not hard to verify that
\begin{align*}
  \sum_{\sigma \in S_m}(-1)^{\sign(\sigma)} \prod_{t=1}^m A_{t,\kappa(\sigma(t))}
= \begin{cases}
  (-1)^{\sign(\kappa)} \det(A), &\text{ if $\kappa$ is a permutation}, \\
  0, &\text{ otherwise}.
\end{cases}
\end{align*}
Thus,
\begin{align*}
  \Omega(A \star q)
= \det(A) \cdot \sum_{\kappa\in S_m} (-1)^{\sign(\kappa)} \left( A \star \frac{\partial q}{\partial Z_{\kappa(1),1} \cdots \partial Z_{\kappa(m),m}} \right)
= \det(A) \cdot \left( A \star \Omega(q) \right).
\end{align*}
This completes the proof.
\end{proof}

As a corollary, we obtain that for any homogeneous polynomial~$q$ of degree $rm$, the map $A\mapsto \Omega^r(A \star q)$ is an $\SL(m)$-invariant (recall that in this case $\Omega^r(A \star q)$ is a complex number).
More generally, the following holds:

\begin{corollary}{\label{Omega_process}} Let $q$ be a homogeneous polynomial of degree $rm$ in the variables $(Z_{i,j})_{i,j=1}^m$ and $A \in \GL(m)$.
Then,
\begin{align*}
\Omega^r(A \star q) = \det(A)^r \cdot \Omega^r(q).
\end{align*}
\end{corollary}

Now let us see how to get a Reynolds operator for a polynomial action of $\SL(m)$ via Cayley's Omega-process.
We assume that $\SL(m)$ acts by homogeneous polynomials of degree~$\ell$ on an $n$-dimensional vector space~$V$.
This means that $Z \circ v = \rho(Z)v$, where $\rho(Z)$ is an $n$ by $n$ matrix whose entries are homogeneous polynomials of degree~$\ell$ in the variables $Z_{i,j}$ (this action is not to be confused with the $\star$ action defined above).

Let $p \in \mathbb{C}[V]$ be a polynomial function on~$V$.
Then we define the polynomial function $Z \circ p$ on~$V$ by
\begin{align*}
(Z \circ p)(v) := p\bigl(\rho(Z)^T v\bigr).
\end{align*}
This defines a polynomial action of $\SL(m)$ on $\C[V]$.
If $p$ is a homogeneous polynomial, we can write it as $p = \sum_{\nu} p_\nu \cdot \nu$, where $p_\nu\in\C$ and $\nu$ ranges over all monomials of the same degree as $p$.
Now, if $p$ is a homogeneous polynomial of degree $tm$ then so is $Z \circ p$, and hence can
be written as
\begin{align*}
Z \circ p = \sum_{\nu} \wt{p}_\nu(Z) \cdot \nu,
\end{align*}
where it is easily checked that the $\wt{p}_\nu(Z)$ are homogeneous polynomials of degree $\ell{}tm$ in the $Z_{i,j}$'s
(which furthermore depend linearly on $p$).
We can now define a scaled version of the Reynolds operator as follows:
\begin{equation}\label{eq:reynolds-slice}
\cR(p) := \Omega^{\ell{}t}(Z \circ p) := \sum_{\nu} \Omega^{\ell t}(\wt{p}_\nu(Z)) \cdot \nu .
\end{equation}
(Recall that $\Omega^{\ell t}(\wt{p}_\nu(Z)) \in\C$.)
We also define $\cR(p):=0$ if $p$ is homogeneous of a degree that is not divisible by $m$,
and extend $\cR$ to a linear map $\C[V] \to \C[V]$.

The next lemma shows that this map $\cR$, suitably rescaled, is a map onto the invariant ring $\C[V]^G$ that satisfies the first defining property of the Reynolds operator.
This is sufficient to ensure \cref{span_inv}.
The second property also holds (but we refrain from verifying it), so that $ c_{\ell t,m}^{-1}\cR$ is indeed the Reynolds operator in our setting.

\begin{lemma}\label{le:Cay-prop}
\begin{enumerate}
\item $c_{\ell t,m} :=\Omega^{\ell t}\big(\det(Z)^{\ell t} \big)$ is a nonzero constant.
\item For all $p \in \mathbb{C}[V]$, $\cR(p)$ is an $\SL(m)$-invariant polynomial.
\item For all $p \in \mathbb{C}[V]^{\SL(m)}$ we have $\cR(p) = c_{\ell t,m} \, p$.
\end{enumerate}
\end{lemma}

\begin{proof}
\begin{enumerate}

\item This is~\cite[Lemma 4.5.25]{Derksen-Kemper}.

\item We will prove that for $Y\in\SL(m)$,
$$
Y \circ \cR(p) = \det(Y)^{\ell t} \cdot \cR(p) ,
$$
by showing that both sides of the equation above are equal to
$\Omega^{\ell t}((YZ) \circ p)$.
For this we may assume that $p$ is homogeneous of degree $tm$.

By decomposing
\begin{equation*}
  (YZ) \circ p = Y \circ (Z \circ p) = Y \circ \sum_{\nu} \wt{p}_\nu(Z)  \cdot \nu
  = \sum_{\nu} \wt{p}_\nu(Z)  \cdot (Y \circ \nu)
\end{equation*}
we get 
\begin{align*}
  \Omega^{\ell t}((YZ) \circ p) &=
  \Omega^{\ell t}\left( \sum_{\nu} \wt{p}_\nu(Z)  \cdot (Y \circ \nu) \right)
  = \sum_{\nu} \Omega^{\ell t}(\wt{p}_\nu(Z))  \cdot (Y \circ \nu) \\
  &= Y \circ \sum_{\nu} \Omega^{\ell t}(\wt{p}_\nu(Z))  \cdot \nu = Y \circ \cR(p) ,
\end{align*}
where the second equality is true because $Y$ and $\nu$ are constants with respect
to $\Omega$, and the third equality holds because of linearity of the group action.
On the other hand, if we decompose 
$$
  (YZ) \circ p = \sum_{\nu} \wt{p}_\nu(YZ) \cdot \nu =
  \sum_{\nu} (Y^T \star \wt{p}_\nu(Z)) \cdot \nu ,
$$
we obtain that $\Omega^{\ell t}((YZ) \circ p)$ equals
\begin{align*}
  \Omega^{\ell t}((YZ) \circ p)
  = \sum_{\nu} \Omega^{\ell t}(Y^T \star \wt{p}_\nu(Z)) \cdot \nu
  = \sum_{\nu} \det(Y)^{\ell t} \cdot \Omega^{\ell t}(\wt{p}_\nu(Z)) \cdot \nu
  = \det(Y)^{\ell t} \cdot \cR(p)
\end{align*}
where the second equality holds due to \cref{Omega_process}.

\item Let $p$ be a nonzero homogeneous invariant. It is easy to see that its degree
is a multiple of $m$, say $tm$.
If $p$ is homogeneous of degree $tm$ and invariant under the
action of $\SL(m)$, there is some integer $d$ such that
$Z \circ p = \det(Z)^{d} \cdot p$.
Since $Z \circ p$ is of degree $\ell tm$ in $Z$, we must have $d=\ell t$.
By the definition of the operator $\cR$, this implies
\[
\cR(p) = \Omega^{\ell t}(Z \circ p) = \Omega^{\ell t}\left(\det(Z)^{\ell t} \right) \cdot p = c_{\ell{}t,m} p. \qedhere
\]
\end{enumerate}
\end{proof}

Now that we have a description of the Reynolds operator $\cR$ for the $\SL(m)$-action, we can use it to bound the size of the coefficients of invariant polynomials $\cR(p)$ obtained by applying the Reynolds operator to some polynomial~$p \in \C[V]$. 

To do that, we will need the following property of the Omega process, which tells us that
applying the Omega process to a polynomial does not increase the size of the coefficients
by much.

\begin{lemma}[Omega process and coefficient size]\label{lem:omega-coeff}
  Let $Z = (Z_{i,j})_{i,j=1}^m$ be a matrix of variables and $p(Z) \in \C[Z]$ be a polynomial of degree $d$ with integer coefficients whose absolute values are bounded by~$R$.
  Then the coefficients of $\Omega^s(p)$ have absolute value bounded by $R \,(d^m \,m!)^s$.
\end{lemma}
\begin{proof}
It suffices to prove the assertion for $s=1$.
Write $p(Z) = \sum_{\alpha\in\N^{m\times m}} p_\alpha Z^\alpha$ with $p_\alpha\in\Z$ such that $\lvert p_\alpha\rvert \le R$; $Z^\alpha$ denotes the monomial $\prod_{i,j} Z_{i,j}^{\alpha_{i,j}}$.
By definition of the $\Omega$-operator,
$$
 \Omega(p) = \sum_\alpha p_\alpha \sum_{\sigma\in S_m}  (-1)^{\sign(\sigma)}
   \left( \prod_{i=1}^m \alpha_{i,\sigma(i)} \right) Z^{\alpha -M_\sigma} ,
$$
where $M_\sigma$ denotes the permutation matrix of $\sigma\in S_m$;
we use the convention that $Z^{\beta} := 0$ if $\beta\not\in\N^{m\times m}$.
Therefore,
$\Omega(p) = \sum_\beta q_\beta Z^\beta$,
with the coefficient $q_\beta$ described as
$$
 q_\beta = \sum_{\sigma}  (-1)^{\sign(\sigma)} p_\alpha \prod_{i=1}^m \alpha_{i,\sigma(i)} ,
$$
where we have abbreviated $\alpha = \beta + M_\sigma$ in the sum.
Since $p_\alpha=0$ if $\sum_{i,j}\alpha_{i,j}  >d$, we obtain
$|q_\beta| \le m!\,R\, d^m$ as claimed.
\end{proof}

%


\begin{proposition}[Reynolds operator and coefficient size]\label{prop:coeff-bounds}
  Let $V = \C^n$ and $\rho\colon\SL(m) \to \GL(V)$ be a polynomial representation, where each $\rho_{i,j}(Z)$ is a homogeneous polynomial of degree~$\ell$ whose coefficients are integers with absolute value at most~$R$.
  In addition, let $p \in \C[V]$ be a homogeneous polynomial of degree~$tm$ such that its coefficients have absolute value at most~$M$.

  Then the coefficients of the invariant polynomial $\cR(p) = \Omega^{\ell t}(Z \circ p)$ are bounded in absolute value by
  $M \left( Rn^2 \right)^{tm} \left(\ell t m^4 \right)^{\ell t m}$.
\end{proposition}
\begin{proof}
  Let $v = (v_1, \ldots, v_n)$ be a vector of $n$~variables. 
  Write $p=\sum p_\nu \cdot \nu$ and $Z \circ p = \sum_\nu \wt{p}_\nu(Z) \cdot \nu$, where $\nu$ ranges over all monomials of degree~$tm$ in the variables $v$.
  According to \cref{eq:reynolds-slice},
  $\cR(p) = \sum_\nu \Omega^{\ell t}(\wt{p}_\nu(Z)) \cdot \nu$,
  so we shall focus on each of the polynomials~$\wt{p}_\nu(Z)$.

  Fix a monomial~$\nu(v) = \prod_{i=1}^n v_i^{\alpha_i}$ of degree~$tm$.
  By applying the linear transformation $\rho(Z)^T$ to~$v$ we obtain that $\nu(\rho(Z)^T v) = \prod_{i=1}^n \left( \sum_{j=1}^n \rho_{j,i}(Z) v_j \right)^{\alpha_i}$.
  This implies that each $\nu(\rho(Z)^T v)$ can be written as a sum of at most~$n^{tm}$ products of the form $\prod_{r=1}^{tm} \rho_{j_r, i_r}(Z) v_{j_r}$.
  As there are exactly $\binom{n-1+tm}{tm}$ monomials~$\nu$ of degree~$tm$ in $n$~variables,
  we have that $(Z \circ p)(v)$ can be written as a sum of at most $\binom{n-1+tm}{tm} n^{tm}$ products of the form
  $p_\nu \prod_{r=1}^{tm} \rho_{j_r, i_r}(Z) v_{j_r}$.
  In particular, the above implies that for each monomial $\nu$ of $Z \circ p$, the coefficient~$\wt{p}_\nu(Z)$ can be written as a sum of at most $\binom{n-1+tm}{tm} n^{tm} \leq n^{2tm}$ polynomials of the form $p_\nu \prod_{r=1}^{tm} \rho_{j_r, i_r}(Z)$.

  Since each polynomial~$\rho_{j,i}(Z)$ is a polynomial of degree~$\ell$ in $m^2$~variables whose coefficients are integers bounded by~$R$, we have that the product $\prod_{r=1}^{tm} \rho_{j_r, i_r}(Z)$ is a polynomial of degree~$\ell tm$ in $m^2$~variables whose coefficients are bounded by $R^{tm} \binom{m^2-1+\ell}{\ell}^{tm} \leq R^{tm} m^{2\ell tm}$.
  Since we assumed that $\lvert p_\nu\rvert \leq M$, we have that $p_\nu \prod_{r=1}^{tm} \rho_{j_r, i_r}(Z)$ is a polynomial of degree $\ell tm$ whose coefficients are integers bounded by $M R^{tm} m^{2\ell tm}$.

  Thus, each $\wt{p}_\nu(Z)$ is a polynomial of degree $\ell{}tm$, with integer coefficients whose absolute values are bounded by
  $M (n^2R)^{tm} m^{2\ell tm}$.
  Now, applying $\Omega^{\ell t}$ to the polynomial $\wt{p}_\nu(Z)$ and using \cref{lem:omega-coeff}, we see that $\Omega^{\ell t}(\wt{p}_\nu(Z))$ has absolute value bounded by
  \begin{equation*}
    M (n^2R)^{tm} m^{2\ell tm} \cdot \left((\ell tm)^m m!\right)^{\ell t}
    \leq M \left( Rn^2 \right)^{tm} \left(\ell t m^4 \right)^{\ell t m}  \qedhere
  \end{equation*}
\end{proof}


\Cref{prop:coeff-bounds} shows in particular that the Reynolds operator for a polynomial action of~$\SL(m)$ of degree~$\ell$ on some vector space $V$ of dimension~$n$, when applied to a monomial, does not generate polynomials with coefficients of high bit complexity.

\subsection{Coefficient and capacity bounds}\label{subsec:coeff-bds}

We now consider a representation $\rho \colon G \to \GL(V)$ of $G = \SL(n_1) \times \cdots \times \SL(n_d)$ of dimension $\dim(V) = n$, where each entry~$\rho_{i,j}$ is a homogeneous polynomial of degree~$\ell$.
By~\cref{lem:reynolds-composition}, together with the fact that $G$ has $N = \SL(n_1) \times \{I_{n_2}\} \times \cdots \times \{I_{n_{d}}\}$ as a normal subgroup, an easy induction shows that
\begin{equation}\label{eq:reynolds}
  \cR_G = \cR_{\SL(n_1)} \cR_{\SL(n_2)} \cdots \cR_{\SL(n_{d})},
\end{equation}
where $\cR_{\SL(n_k)}$ denotes the Reynolds operator corresponding to the action of the $k$-th factor, i.e., $\rho^{(k)}(Z) = \rho(I_{n_1}, \dots, I_{n_{k-1}}, Z, I_{n_{k+1}}, \dots, I_{n_d})$.
To apply the results of the previous section, we will assume that each entry $\rho^{(k)}_{i,j}$ is given by a homogeneous polynomial of degree no larger than~$\ell_k$, with integer coefficients bounded in absolute value by some constant~$R$.
Note that $\sum_{k=1}^d \ell_k=\ell$.

We are now ready to prove the main results of \cref{sec:cayley}:
upper bounds on the \emph{coefficients} of the invariant polynomials and, as a corollary, lower bounds on the \emph{capacity} $\capac(v) := \inf_{g\in G} \lVert \rho(g) v\rVert_2^2$ for vectors that are not in the null cone.
In particular, we obtain alternative bounds for the tensor action of $G$ on $\Ten(n_0,\ldots, n_{d})$, where $n=n_0 n_1\cdots n_d$, $\ell=d$, $\ell_k=1$ for all $k\in[d]$, and $R=1$.

Thus, applying \cref{prop:coeff-bounds} to each of the Reynolds operators $\cR_{\SL(n_i)}$, and noting that each of these operators has the
form~\eqref{eq:reynolds-slice}, we obtain the main result of this section:

\begin{theorem}\label{thm:coefficient-bounds}
  The space $\C[V]^G_D$ of homogeneous invariants of degree $D$ is spanned by invariants with integer coefficients
        of absolute value at most~$L$, where
  \begin{align*}
   \ln L =O\left( dD \ln(Rn) +  \ell D \ln \Big(D \max_{i\ge 1} \ell_i n_i\Big) \right) .
  \end{align*}
  For the tensor action of~$G$ on $\Ten(n_0,\ldots, n_d)$, this bound becomes $O(d D \ln(D n))$.
\end{theorem}
\begin{proof}
  By \cref{span_inv}, the invariants of degree $D$ are spanned by the polynomials $\cR_G(\nu)$, where $\nu$ ranges over
  all monomials of degree $D$ in $V$.
  By \cref{eq:reynolds} and \cref{prop:coeff-bounds}, we find that the coefficients of $\cR_G(\nu)$ are bounded by
  \begin{align*}
    \prod_{i=1}^d \left[ (Rn^2)^D \cdot (\ell_i D n_i^3 )^{\ell_i D} \right]
    &\leq \exp\left( O\left( d D \ln(Rn) + \sum_{i=1}^d \ell_i D \ln(D \ell_i n_i) \right) \right)\\
    &\leq \exp\left( O\left( d D \ln(Rn) + \ell D \ln(D \max_{i\geq1} \ell_i n_i) \right) \right).
  \end{align*}
  For the tensor action, use that $\ell_i=1$, $\ell=d$, $R=1$, and $n=n_0 n_1\cdots n_d$.
\end{proof}


Note that these bounds on the coefficients of the generators are somewhat worse than the bounds obtained in~\cref{subsec:invariants}.
However, one advantage of the method present in this chapter is that the bounds hold for more general actions than just the tensor product action of $G$ on $\Ten(n_0,\ldots, n_{d})$.
Moreover, it allows us to compute generators for the invariant rings of more general actions.

The null cone is in fact cut out already by a subset of the generators of the invariant ring.
Their degree can be bounded using a result by Derksen, which we instantiate here
according to our parameters.

\begin{lemma}[{\cite[Proposition~1.2]{derksen2001polynomial}}]\label{lem:degree bound}
The null cone of the $G$-representation~$V$ is the zero set of homogeneous $G$-invariant polynomials of degree
\begin{align*}
D
\leq \ell^{\sum_{i=1}^d (n_i^2 - 1)} \cdot \max_{i\ge 1} n_i^d
\leq \exp(2 d \ln(\ell) \cdot \max_{i\ge 1} n_i^2 ).
\end{align*}
For the tensor action of $G$ on $\Ten(n_0,\ldots, n_{d})$, the bound becomes $D \leq \exp(2 d \ln(d) \cdot \max_{i\ge 1} n_i^2).$
\end{lemma}

\begin{proof}[Proof]
  This is a direct application of~\cite[Proposition 2.1]{derksen2001polynomial}, with parameters
  $A = \ell$,
  $H = \max_{i\in[d]} n_i$,
  $t = \sum_{i=1}^d n_i^2$ and $m = \sum_{i=1}^d (n_i^2-1)$.
   For the tensor action, the representation is given by homogeneous polynomials of degree $\ell = d$.
\end{proof}

In particular, we may obtain a complete set of polynomials defining the null cone by taking all invariants from the preceding theorem up to the degree predicted by Derksen.

We now provide a variant for \cref{prp:inf norm on orbit} for general homogeneous $G$-representations\footnote{Using this lemma, the
logarithmic dependence on $\max_k n_k$ in \cref{thm:main precise}
changes into a polynomial dependence, but it doesn't matter for the time complexity of the
resulting algorithm; see the comments after that theorem.}.

\begin{theorem}
Let $V$ be a representation of $G$ of dimension $n$ and $v \in V$ be a vector with
integral entries that is not in the null cone of the $G$ action.
Then the capacity of $v$ is bounded below by
\begin{align*}
  \capac(v) = \inf_{g\in G} \lVert \rho(g) v\rVert_2^2 \geq
   \exp(-c \left( d\ln(Rn) + d \ell \ln(\ell) \max_{i\ge 1} \ell_i n_i^2 \right) ),
\end{align*}
where $c >0$ is a universal constant.
For the tensor action of $G$ on $\Ten(n_0, \ldots, n_d)$, we have a bound of the form
$\exp\left(-c \; d^2 \ln(d) \max_{i \geq 0} n_i^2 \right)$.
\end{theorem}
\begin{proof}
According to \cref{lem:degree bound} and \cref{thm:coefficient-bounds}, there is a homogeneous polynomial $P$
of degree~$D$ and integer coefficients with absolute value bounded by $L$ such that $P(v)\ne 0$ and
$$
 \ln D \le 2 d \ln(\ell) \max_{i\ge 1} n_i^2, \quad
 \ln L \le c' \big( dD\ln(Rn) +  \ell D \ln\big(D \max_{i\ge 1} \ell_i n_i\big) \big) ,
$$
where $c' > 0$ is a universal constant.
For $g \in G$ we have
\[ 1 \leq \lvert P(v) \rvert = \lvert P(g \cdot v) \rvert
\leq n^D L \lVert g \cdot v \rVert_\infty^D
 \leq n^D L \lVert g \cdot v \rVert_2^D , \]
and hence, for another universal constant $c>0$,
\begin{align*}
\lVert g \cdot v \rVert_2 \geq \frac 1 {n} L^{-\frac1D}
&= \frac 1 {n} \exp\left( -c' \bigg( d \ln(Rn) + \ell \ln\big(D \max_{i\ge 1} \ell_i n_i\big) \bigg) \right)\\
&\ge  \exp\left( -\frac c2 \bigg( d \ln(Rn) + d\ell \ln(\ell) \max_{i\ge 1} \ell_i n_i^2 \right) .
\end{align*}
For the tensor action, again use that in this case $\ell_i=1$, $\ell = d$, $R = 1$,
and $n = n_0 n_1 \cdots n_d$.
\end{proof}


\subsection{Algebraic algorithm for testing membership in the null cone}

We are now ready to describe an algebraic algorithm for the null-cone membership problem
for tensors, complementing our previous \cref{thm:null cone exp scaling}.
Our approach is similar to the one in~\cite[Section 4.6]{Sturmfels}, in that we use
the Reynolds operator applied to the monomials to obtain a generating set for $\C[V]^G$.
However, we improve on the algorithm in~\cite{Sturmfels} by using Derksen's exponential
bounds on the degree of the generating invariants and by avoiding the use of
Gr\"{o}bner bases. Thus we obtain an exponential time algorithm, as opposed to the
doubly-exponential time algorithm in~\cite{Sturmfels}.

\begin{Algorithm}
\textbf{Input}: A tensor $X$ in $\Ten(n_0,\ldots, n_d)$ with integer
coordinates (specified as a list of coordinates, each encoded in binary with
bit complexity at most~$b$). \\
\textbf{Output}: Is $X$ in the null cone? \\
\textbf{Algorithm}:
Let $D = \exp(2 d \ln(d) \cdot \max_{i\ge 1} n_i^2)$ be the degree bound of~\cref{lem:degree bound}.
\begin{enumerate}
  \item For each monomial $\nu$ 
        of degree at most~$D$, obtain $\RO(\nu)$ by applying the Reynolds operator from~\eqref{eq:reynolds}.
  Proceed to the next step.
  \item For each $\nu$ as above, compute $\RO(\nu)(X)$. If $\RO(\nu)(X) \neq 0$ output
  that $X$ is not in the null cone.
  \item Otherwise, if $\RO(\nu)(X) = 0$ for all $\nu$ in the previous step, output
  that the tensor $X$ is in the null cone.
\end{enumerate}
\caption{Algorithm for the null-cone problem}\label{alg:null cone alternative}
\end{Algorithm}

\begin{theorem}\label{thm:exp time via omega and derksen}
\cref{alg:null cone alternative} decides
  the null-cone problem for the action of $G = \SL(n_1)\times\dots\times\SL(n_{d})$ on
  $\Ten(n_0,\dots,n_d)$ in time
$$
  \poly(b) \exp(O\Big(n d \ln(d) \max_{i\ge 1} n_i^2\Big)),
$$
  where $b$ is an upper bound on the bitsize of the coordinates of the input tensor $X$ (specified as a
  list of coordinates, each encoded in binary).
\end{theorem}


\begin{proof}
We prove the correctness of~\cref{alg:null cone alternative} and analyze its running time.
By~\cref{span_inv} and Derksen's bound~\cref{lem:degree bound},
we know that the null cone is the zero set of
$\{ \RO(\nu) \mid \nu \text{ monomial of degree} \leq D \}$.
Hence the algorithm is correct, as it will evaluate $X$ on a set of polynomials
that cuts out the null cone.

Now, we are left with analyzing the complexity of our algorithm.
By~\cref{thm:coefficient-bounds}, we know that each $\RO(\nu)$ is a polynomial of degree at most~$D$ on
$n := n_0 n_1 \cdots n_d$ variables
with integer coefficients of \emph{bit size} bounded by $L$, where
$$
 \ln(D) = 2 d \ln(d) \cdot \max_{i\ge 1} n_i^2,\quad
 L = O\Big( dD \ln(Dn) \Big) =  D^{O(1)} .
$$
For each monomial $\nu$ of degree at most~$D$, it takes $\poly(B, b, D)$ time to evaluate $\nu$ at point~$X$.
There are at most $D^n$ monomials $\nu$ of degree at most~$D$.
Thus the total time to compute each $\RO(\nu)(X)$ is given by $\poly(B, b, D^n) =\poly(b, D^n) $,
which is bounded by $\poly(b) \cdot \exp(O(n\ln(D)))$.
\end{proof}



\begin{remark}\label{re:alg_gen}
The advantage of the algebraic approach in~\cref{alg:null cone alternative}
presented here is its generality: it immediately
extends to any polynomial action of $G = \SL(n_1)\times\dots\times\SL(n_{d})$
on any vector space. Moreover, an analysis analogous to
\cref{thm:exp time via omega and derksen} holds, which is clear from the construction.
\end{remark}

\section{Conclusion and open problems}\label{sec:conclusion}

This paper continues the recent line of works studying the computational complexity of problems in invariant theory~\cite{Mulmuley,GGOW,IQS2015,IQS15b,ikenmeyer2017vanishing,buergisser2017membership,gct5,baldoni2017computation}.
There are many beautiful algorithms known in invariant theory~\cite{Derksen-Kemper,Sturmfels}, but most of
them come without a complexity analysis or will have at least exponential runtime. 
Designing efficient algorithms for problems in invariant theory is important for the GCT program~\cite{Mulmuley,gct5}, and
we believe that viewing invariant theory from the computational lens will provide significant new structural insights. 
Several open problems arise from this work. We mention the most interesting ones:

\begin{enumerate}
\item Is there a polynomial time algorithm for deciding the null-cone problem for the tensor actions that we study?
It would suffice to give an analog of the algorithm in~\cref{thm:main} with running time $\poly(n, b, \log(1/\eps))$.
One might wonder if known optimization techniques are sufficient to yield this runtime guarantee.
While the optimization problem in \cref{eq:alg-alt-min} is not convex, it is nevertheless known to be \emph{geodesically} convex~\cite{KN79,NM84,Woodward11} (roughly speaking one needs to move from the usual Euclidean geometry to the geometry of the group to witness convexity).
The theory of geodesically convex optimization is just starting to emerge (see, e.g., \cite{Zhang_Sra16} and the references therein) and it is an open problem whether there are computationally efficient analogs of the ellipsoid algorithm and interior point methods (algorithms that in Euclidean convex optimization guarantee $\poly(\log(1/\eps))$ runtime).
\item Is there a polynomial time algorithm for the null-cone problem for more general group actions?
As mentioned before, there is a natural notion of \emph{``$G$-stochasticity''} using moment maps provided by the noncommutative duality theory.
A first important goal is to design an algorithm that achieves the analog of our~\cref{thm:main} (i.e., getting to a point in the orbit where the ``distance" to satisfying $G$-stochasticity is at most $\eps$).
While it is not clear how far the alternating minimization approach can be taken, Kirwan's gradient flow~\cite{Kirwan84} seems to be a promising alternative first proposed in~\cite{walter2013entanglement,walter2014multipartite}.
\item Is there a polynomial time algorithm for the one-body quantum marginal problems described in \cref{sec:quantum}?
There is a natural scaling algorithm generalizing \cref{alg:scaling} but it has so far evaded analysis.
Even obtaining a polynomial time algorithm for a promise problem along the lines of \cref{thm:main} would be rather interesting.
\end{enumerate}

\section*{Acknowledgements}
We would like to thank Henry Cohn, Eric Naslund, Ion Nechita, Will Sawin, Suvrit Sra, Josu\'e Tonelli, and John Wright for helpful discussions.
MW acknowledges financial support by the NWO through Veni grant no.~680-47-459.
PB acknowledges financial support from the DFG grant BU~1371~2-2.

\appendix

\section{Proofs in geometric invariant theory}\label{app:git}

Here we provide some missing proofs from the main part of the paper.
None of the results herein are new and we claim no originality.
We hope that our concrete proofs will be useful to make the paper self-contained and facilitate further developments.

First we provide a proof of the Kempf-Ness theorem~\cite{KN79} in the case of tensors, which we restate for convenience.

\KempfNess*

In one direction, if $X$ is not in the null cone, then the desired tensor $Y$ will be a scalar multiple of a tensor of minimal norm in the $G$-orbit closure of $X$. The proof of $d$-stochasticity will be by acting elements close to identity on $Y$ and noting that this does not decrease the norm. In the other direction, the proof will follow by showing that a tensor $Y$ in $d$-stochastic position is a tensor of minimal norm in its orbit which will follow from applying concavity of the logarithm after an appropriate change of basis.

\begin{proof}
First suppose that $X$ is not in the null cone. Then there exists a nonzero tensor $Y$ of minimal norm in the orbit closure of $X$ (under the action of the group $G$),
$$
Y = \argmin_{Z \in \overline{G \cdot X}} \|Z\|_2 .
$$
We will prove that the quantum marginals $\rho_1,\ldots, \rho_{d}$ of $Y Y^\dagger / Y^\dagger Y$ satisfy that $\rho_i = I_{n_i} / n_i$ for $i\in[d]$.
It suffices to show that the marginals $\tilde\rho_i$ of the unnormalized PSD matrix $Y Y^\dagger$ are proportional to identity matrices.
For this, consider acting on the $i^{\text{th}}$ coordinate by $\expon(t P)$, where $t$ is a parameter and $P$ is an arbitrary $n_i \times n_i$ matrix satisfying $\tr(P) = 0$, so that $\expon(t P) \in \SL(n_i)$.
Note that
$$
\left\|(I_{n_0},\ldots, I_{n_{i-1}}, \expon(t P), \ldots, I_{n_d}) \cdot Y \right\|_2^2 = \tr\left[ \expon(2 t P) \tilde\rho_i\right] .
$$
Since $Y$ is a tensor of minimal norm, acting by the group does not change the norm square to first order, so we obtain that
$$
0 = \frac{d}{dt}  \tr\left[ \expon(2 t P) \tilde\rho_i\right] \Big\rvert_{t = 0} = 2 \tr[P \tilde\rho_i] .
$$
This is true for any traceless matrix $P$.
Hence $\tilde\rho_i$ has to be a scalar multiple of identity.
This proves one direction.

For the other direction,
suppose there exists $Y\in \overline{G\cdot X}$ such that the quantum marginals $\rho_i$ of
$Z := Y / \norm{Y}_2$ satisfy $\rho_i=I_{n_i}/n_i$ for all $i\in [d]$.
It suffices to prove that
$$
\capac(Z) \ge 1 . 
$$
(Then $Z$ is not in the null cone and so neither is $Y$, since the null cone is $G$-invariant and closed.)

Consider $(A_1,\ldots, A_{d})$, where $A_i \in SL(n_i)$. Then
$$
\left\| (A_1,\ldots, A_{d}) \cdot Z\right\|_2^2 = \tr\left[ \left(A_1 A_1^{\dagger} \otimes \cdots \otimes A_{d} A_{d}^{\dagger} \right)  \rho_{1\cdots d} \right] ,
$$
where $\rho_{1\cdots d}$ is the quantum marginal of $Z Z^{\dagger}$ on the last $d$ tensor factors.
Consider the eigenvalue decomposition of each $A_i A_i^{\dagger}$,
$$
A_i A_i^{\dagger} = \sum_{j=1}^{n_i} \expon(\lambda_{i, j}) u_{i, j} u_{i, j}^{\dagger} .
$$
Note that $\sum_{j=1}^{n_i} \lambda_{i,j} = 0$ for all $i\in[d]$ since $\det(A_i)=1$.
Now consider the tensor $T \in \Ten(n_1,\ldots,n_{d})$ given by
$$
T_{j_1,\ldots, j_{d}}
 :=  \tr\left[ \left(u_{1, j_1} u_{1, j_1}^{\dagger} \otimes \cdots \otimes u_{d, j_{d}} u_{d, j_{d}}^{\dagger} \right) \rho_{1\cdots d} \right] ,
$$
which has nonnegative real entries.
We claim that they add up to one.
Indeed, for fixed $i\in [d]$ and $j\in [n_i]$, we have the following stronger stochasticity property:
\begin{equation}\label{eq:magic stochastic}
\begin{aligned}
\sum_{j_1=1}^{n_1} \cdots  \sum_{j_{i-1}=1}^{n_{i-1}} \sum_{j_{i+1}=1}^{n_{i+1}}  \cdots  \sum_{j_{d} = 1}^{n_{d}} T_{j_1,\ldots, j_{d}}
 &=  \tr\left[ \left(I_{n_1} \otimes \cdots I_{n_{i-1}} \otimes u_{i, j_i} u_{i, j_i}^{\dagger} \otimes I_{n_{i+1}} \otimes \cdots I_{n_{d}} \right)     \rho_{1\cdots d} \right] \\
&= \tr \left[ u_{i, j_i} u_{i, j_i}^{\dagger} \rho_i\right] = \frac 1 {n_i}.
\end{aligned}
\end{equation}
(The first equality follows from the fact that the $u_{i,j}$'s form an orthonormal basis;
the second equality follows from the definition of a quantum marginal;
the third equality follows from the fact that the $u_{i,j}$'s have norm~$1$ and that $\rho_i = I_{n_i}/n_i$.)
Summing over $j\in [n_i]$ shows the claim.

We can now write the quantity we want to lower bound as follows:
$$
\tr\left[ \left(A_1 A_1^{\dagger} \otimes \cdots \otimes A_{d} A_{d}^{\dagger} \right) \rho_{1\cdots d} \right]
   = \sum_{j_1=1}^{n_1} \cdots \sum_{j_{d} = 1}^{n_{d}} \expon\left(\sum_{i=1}^{d} \lambda_{i, j_i}\right) \cdot T_{j_1,\ldots, j_{d}}
$$
Using the concavity of the logarithm, we obtain
\begin{align*}
&\: \: \: \: \log \Bigg(  \sum_{j_1=1}^{n_1} \cdots \sum_{j_{d} = 1}^{n_{d}} \expon\left(\sum_{i=1}^{d} \lambda_{i, j_i}\right) \cdot T_{j_1,\ldots, j_{d}} \Bigg)
\ge \sum_{j_1=1}^{n_1} \cdots \sum_{j_{d} = 1}^{n_{d}} \left(\sum_{i=1}^{d} \lambda_{i, j_i}\right) \cdot T_{j_1,\ldots, j_{d}} \\
&= \sum_{i=1}^{d} \sum_{j_i = 1}^{n_i} \lambda_{i, j_i} \cdot \sum_{j_1=1}^{n_1} \cdots  \sum_{j_{i-1}=1}^{n_{i-1}} \sum_{j_{i+1}=1}^{n_{i+1}}  \cdots  \sum_{j_{d} = 1}^{n_{d}} T_{j_1,\ldots, j_{d}} = 0 .
\end{align*}
(The first equality is just a rearrangement of sums;
the second equality follows from the stochasticity property~\eqref{eq:magic stochastic} established above
and the fact that $\sum_{j=1}^{n_i} \lambda_{i,j} = 0$ for all $i\in[d]$.)
This completes the proof.
\end{proof}

Next we prove the special case of the Hilbert-Mumford criterion~\cite{Hil, Mum65}:

\HilbertMumford*

One direction ($\supp(Y)$ being deficient implies tensor $X$ is in the null cone) is easy and already outlined in the main part of the manuscript.
The proof of the other direction will proceed by choosing an appropriate basis, reducing the problem to LP feasibility, and using LP duality.

\begin{proof}
First suppose that there exist unitary matrices $U_i$ such that the support of $Y = (U_1,\ldots, U_{d}) \cdot X$ is deficient.
We will prove that $Y$ is in the null cone which will imply that $X$ is in the null cone\footnote{Note that unitary matrices may not lie in the special linear group but some scalar multiple of unitary matrices does. This suffices for our purposes.}.
Since $\supp(Y)$ is deficient, there exist integers $a_{i,j}$, $i\in[d]$, $j\in[n_i]$, such that
\begin{align*}
\forall \: i\in[d]\colon \quad &\sum_{j=1}^{n_i} a_{i,j} = 0,  \\ 
\forall \: (j_1,\ldots, j_{d}) \in \supp(Y)\colon \quad&\sum_{i=1}^{d} a_{i, j_i} > 0 .
\end{align*}
Then the action of the group element
$\left(\diag\left( z^{a_{1, 1}}, \ldots, z^{a_{1, n_1}}\right), \ldots, \diag\left( z^{a_{d, 1}}, \ldots, z^{a_{d, n_{d}}}\right)\right)$
sends $Y$ to~$0$ as $z \rightarrow 0$ (see \cref{eq:action of 1-PSG}) and hence $Y$ is in the null cone.

In the other direction, suppose that $X$ is in the null cone. Then there exists a sequence of group elements
$g^{(1)}, \ldots, g^{(k)}, \ldots$ such that $\lim_{k \rightarrow \infty} g^{(k)} \cdot X = 0$.
Let $g^{(k)} = \left(V^{(k)}_1 D^{(k)}_1 U^{(k)}_1, \ldots, V^{(k)}_{d} D^{(k)}_{d} U^{(k)}_{d}\right)$, where we have performed a singular value decomposition (SVD) for each component of $g^{(k)}$.
Note that since the action by unitary matrices doesn't change the norm of a tensor, the sequence of group elements $\left(D^{(k)}_1 U^{(k)}_1, \ldots, D^{(k)}_{d} U^{(k)}_{d}\right)$ also sends $X$ to $0$.
We can w.l.o.g.\ assume that the sequence $\left(U^{(k)}_1, \ldots, U^{(k)}_{d}\right)$ converges to a sequence of unitary matrices $\left(U_1, \ldots, U_{d}\right)$
(otherwise we can focus on a convergent subsequence, which exists because of the Bolzano-Weierstrass theorem; it applies because the space of tuples of unitary matrices is compact).
This limit $\left(U_1, \ldots, U_{d}\right)$ will be the tuple which we are after.
Now, we have that the sequence $\left(D^{(k)}_1 U_1, \ldots, D^{(k)}_{d} U_{d}\right)$ also sends $X$ to $0$.
Let $Y = (U_1,\ldots, U_{d}) \cdot X$.
We would like to prove that $\supp(Y)$ is deficient and we know that
$\lim_{k \rightarrow \infty}\left(D^{(k)}_1, \ldots, D^{(k)}_{d}\right) \cdot Y = 0$, where
$D^{(k)}_i$ are diagonal matrices with positive entries and determinant equal to $1$.
Suppose $D^{(k)}_i = \diag\left( d^{(k)}_{i, 1}, \ldots, d^{(k)}_{i, n_i}\right)$.
Then the condition $\lim_{k \rightarrow \infty}\left(D^{(k)}_1, \ldots, D^{(k)}_{d}\right) \cdot Y = 0$ is equivalent to
\begin{equation}\label{eqn:ankit1}
\forall \: (j_1, \ldots, j_{d}) \in \supp(Y)\colon  \quad\lim_{k \rightarrow \infty} \prod_{i=1}^{d} d^{(k)}_{i, j_i} = 0,
\end{equation}
Now we are going to essentially prove the Hilbert-Mumford criterion for the above commutative group action using Farkas' lemma (or linear programming duality).
We first note that \cref{eqn:ankit1} implies that there is no tensor $T \in \Ten(n_1,\ldots,n_{d})$ with nonnegative real entries and $\supp(T) \subseteq \supp(Y)$ such that, for all $i\in[d]$,
$$
\sum_{j_1=1}^{n_1} \cdots  \sum_{j_{i-1}=1}^{n_{i-1}} \sum_{j_{i+1}=1}^{n_{i+1}}  \cdots  \sum_{j_{d = 1}}^{n_{d}} T_{j_1,\ldots, j_{d}} = 1/n_i
$$
(the classical version of being in $d$-stochastic position).
Indeed, suppose that, on the contrary, there was such a tensor $T$.
Then consider the following sequence of sums
$$
S^{(k)} := \sum_{j_1=1}^{n_1} \cdots \sum_{j_{d} = 1}^{n_{d}} T_{j_1,\ldots, j_{d}} \: \prod_{i=1}^{d} d^{(k)}_{i, j_i} .
$$
One one hand, we can prove that $S^{(k)} \ge 1$ by the same argument as in the proof of~\cref{thm:kempf ness} above.
On the other hand $\lim_{k \rightarrow \infty} S^{(k)} = 0$ because of \eqref{eqn:ankit1} and the fact that $\supp(T) \subseteq \supp(Y)$.
This means that the following linear program has no solution:
\begin{align*}
\forall \: i\in[d], j\in[n_i]\colon \quad &
 \sum_{j_1=1}^{n_1} \cdots  \sum_{j_{i-1}=1}^{n_{i-1}} \sum_{j_{i+1}=1}^{n_{i+1}}  \cdots  \sum_{j_{d} = 1}^{n_{d}} T_{j_1,\ldots, j_{d}} = 1/n_i, \\
\forall \: j_1\in[n_1],\ldots, j_{d}\in[n_d]\colon \quad & T_{j_1,\ldots, j_{d}} \ge 0,  \\
\forall \: ( j_1,\ldots, j_{d}) \notin \supp(Y)\colon \quad & T_{j_1,\ldots, j_{d}} = 0
\end{align*}
By Farkas' lemma (e.g., \cite{schrijver1998theory}), there exist real numbers $a_{i, j}$, $i\in[d]$, $j\in[n_i]$, such that
$$
\forall \: (j_1,\ldots, j_{d}) \in \supp(Y)\colon \quad \sum_{i=1}^{d} a_{i, j_i} \ge 0 \quad \mbox{and}
  \quad \sum_{i=1}^{d} \sum_{j=1}^{n_i} \frac {a_{i, j}} {n_i} < 0 .
$$
Let us denote $\operatorname{avg}_i := \sum_{j=1}^{n_i} \frac {a_{i, j}} {n_i}$.
We construct the new solution
$\widetilde{a}_{i,j} := a_{i,j} - \operatorname{avg}_i$, which satisfies
\begin{align*}
\forall \: i\in[d]\colon \quad & \sum_{j=1}^{n_i} \widetilde{a}_{i,j} = 0 ,\\ 
\forall \: (j_1,\ldots, j_{d}) \in \supp(Y)\colon \quad &\sum_{i=1}^{d} \widetilde{a}_{i, j_i} > 0 .
\end{align*}
Since there is a real solution with the above properties, an integer solution exists as well.
This proves that $\supp(Y)$ is deficient,
which completes the proof.
\end{proof}

We now show that the instability can also be defined by optimizing over all tuples of invertible (rather than unitary) matrices.

\insSvsU*
\begin{proof}
We first observe that in the form of the 1-PSG in \cref{eq:1-PSG}, as well as when evaluating the instability, we may assume that the numbers $a_{i,j}$ are ordered non-increasingly for any fixed~$i$, i.e., that the tuple $(a_{i,j})$ is an element of
\[
  \Gamma_+ = \Bigl\{ (a_{i,j})_{i\in[d],j\in[n_k]} \;\Big|\; a_{i,j} \in \Z, \, \sum_{j=1}^{n_i} a_{i,j}=0 \text{ and } a_{i,1}\geq a_{i,2}\geq\dots\geq a_{i,n_i} \text{ for } i\in[d] \Bigr\}.
\]
This is because we can always permute the entries of $a_{i,j}$ for each $i$ by the action of appropriate permutation matrices, which we can absorb into the~$B_i$.

Now recall that any invertible matrix has a QR decomposition.
For each $i\in[d]$, let $B_i^{-1} = U_i^\dagger R_i$ be the QR-decomposition of the inverse of $B_i$, i.e., $U_i$ is unitary and $R_i$ upper triangular.
Then $B_i = R_i^{-1} U_i$ for $i\in[d]$.

In order to prove the lemma, it therefore suffices to show that
\[ \defi(\supp(R B \cdot X)) \geq \defi(\supp(B \cdot X)) \]
for any tuple $R=(R_1,\dots,R_{d})$ of upper triangular matrices (these form a group, and so this in fact implies that the inequality is tight).
Thus let $(a_{i,j})\in\Gamma_+$.
Then,
\begin{align*}
\left(R B \cdot X\right)_{k_0,k_1, \ldots, k_d}
= \sum_{(j_1,\ldots,j_{d})\in\supp(B \cdot X)} \left(B \cdot X\right)_{k_0,j_1, \ldots, j_{d}} \prod_{i=1}^{d} \left(R_i\right)_{k_i, j_i} \\
= \sum_{\substack{(j_1,\ldots,j_{d})\in\supp(B \cdot X)\\ j_i\geq k_i}} \left(B \cdot X\right)_{k_0,j_1, \ldots, j_{d}} \prod_{i=1}^{d} \left(R_i\right)_{k_i, j_i},
\end{align*}
where in the first step we used the definition of the support and in the second that the $R_i$ are upper triangular.
This implies that if $(k_1,\dots,k_{d}) \in \supp(R \cdot B \cdot X)$ then there exists $(j_1,\dots,j_{d})\in\supp(B \cdot X)$ with $j_i\geq k_i$ for $i\in[d]$.
Since $(a_{i,j})\in\Gamma_+$, it follows that
\[ \sum_{i=1}^{d} a_{i,k_i} \ \geq\ \sum_{i=1}^{d} a_{i,j_i} , \]
therefore,
\[ \min_{(k_1,\dots,k_{d}) \in \supp(R B \cdot X)}\left(\sum_{i=1}^{d} a_{i,k_i}\right) \ \geq\  \min_{(j_1,\dots,j_{d}) \in \supp(B \cdot X)}\left(\sum_{i=1}^{d} a_{i,j_i}\right). \]
Thus we have showed that, indeed, $\defi(\supp(R B \cdot X)) \geq \defi(\supp(B \cdot X))$, and the lemma follows.
\end{proof}

Next we prove the lemma relating the instability and minimal distance to $d$-stochasticity (a special case of~\cite[Lemma 3.1, (iv)]{NM84}).

\insvsds*
\begin{proof}
According to \cref{lem:instability S vs U}, the instability is $G$-invariant.
Clearly, it is also invariant under rescaling, since this leaves the support of a tensor invariant.
It follows that we only need to prove that
\[ \ins(Y) \leq \sqrt{\ds(Y)} \]
for any tensor $Y\in\Ten(n_0,\dots,n_d)$.
Since $\ds(Y)$ is invariant under the action of tuples of unitary matrices, it in fact suffices to show that
\[ \frac {\min_{(j_1,\dots,j_{d})\in\supp(Z)} \sum_{i=1}^{d} a_{i,j_i}} {\left(\sum_{i=1}^{d} \sum_{j=1}^{n_i} a_{i,j}^2
\right)^{1/2}} \leq \sqrt{\ds(Z)} \]
for any tensor $Z\in\Ten(n_0,\dots,n_d)$ and for any $(a_{i,j})\in\Gamma$.
Since both the instability and $\ds$ are invariant under rescaling, we may we may assume that $\norm{Z}_2^2=Z^\dagger Z=1$.

Now let $\rho_1,\dots,\rho_{d}$ denote the quantum marginals of the last $d$ systems of $ZZ^\dagger$.
Then,
\begin{align*}
& \min_{(j_0,\dots,j_{d})\in\supp(Z)} \sum_{i=1}^{d} a_{i,j_i}
\leq \sum_{i=1}^{d} \sum_{j_1,\dots,j_d} a_{i,j_i}  \abs{Z_{j_1,\dots,j_d}}^2 \\
&= \sum_{i=1}^{d} \tr[\left(I_{n_1} \ot \dots \ot I_{n_{i-1}} \ot {\begin{psmallmatrix}a_{i,j_1}&&\\&\ddots&\\&&a_{i,j_{n_i}}\end{psmallmatrix}} \ot I_{n_{i+1}} \ot \dots \ot I_{n_d} \right) ZZ^\dagger] \\
&= \sum_{i=1}^{d} \tr[{\begin{psmallmatrix}a_{i,j_1}&&\\&\ddots&\\&&a_{i,j_{n_i}}\end{psmallmatrix}} \rho_i]
= \sum_{i=1}^{d} \tr[{\begin{psmallmatrix}a_{i,j_1}&&\\&\ddots&\\&&a_{i,j_{n_i}}\end{psmallmatrix}} \left(\rho_i - \frac {I_{n_i}} {n_i}\right)]
= \sum_{i=1}^{d} \sum_{j=1}^{n_i} a_{i,j} \left(\rho_i - \frac {I_{n_i}} {n_i}\right)_{j,j} \\
&\leq \left( \sum_{i=1}^{d} \sum_{j=1}^{n_i} a_{i,j}^2 \right)^{1/2} \left( \sum_{i=1}^{d} \sum_{j=1}^{n_i} \left(\rho_i - \frac {I_{n_i}} {n_i}\right)_{j,j}^2 \right)^{1/2}
\leq \left( \sum_{i=1}^{d} \sum_{j=1}^{n_i} a_{i,j}^2 \right)^{1/2} \left( \sum_{i=1}^{d} \sum_{\iota,j=1}^{n_i} \left(\rho_i - \frac {I_{n_i}} {n_i}\right)_{\iota,j}^2 \right)^{1/2} \\
&= \left( \sum_{i=1}^{d} \sum_{j=1}^{n_i} a_{i,j}^2 \right)^{1/2} \sqrt{\ds(Z)} ,
\end{align*}
where the first inequality holds because $\norm{Z}^2_2=1$ and so the $\abs{Z_{j_1,\dots,j_d}}^2$ form a probability distribution supported only on indices $(j_0,\dots,j_d) \in [n_0] \times \supp(Z)$,
the third step is by the definition of the quantum marginals (\cref{eq:rdm}),
next we evaluate the trace in the standard basis and we write $(\rho_i)_{j,j}$ for the diagonal entries of $\rho_i$,
then we use the Cauchy-Schwarz inequality,
in the penultimate step we extend the sum over the absolute values squared of all matrix elements rather than the diagonals, and lastly we plug in the definition of $\ds(Z)$ (\cref{def:ds}).
This is exactly the inequality that we needed to prove.
\end{proof}

Now we prove the exponential lower bound on $\dds(X)$ and the instability.

\inslb*
\begin{proof}
The first inequality is \cref{lem:ins vs ds}, so it remains to prove the lower bound on the instability.
If $X$ is unstable then $\ins(X)>0$, i.e., there exists a tuple of unitaries $U=(U_1,\dots,U_{d})$ and $(a_{k,j})\in\Gamma$
such that the support $S := \supp(U \cdot X)$ has the property that
$\min_{(j_1,\dots,j_{d}) \in S}\left(\sum_{i=1}^{d} a_{i,j_i}\right) > 0$, i.e., is at least one.
Thus we have
\[ \ins(X) \geq \defi(S) =
 \max_{(a_{i,j})\in\Gamma}\frac {\min_{(j_1,\dots,j_{d}) \in S}\left(\sum_{i=1}^{d} a_{i,j_i}\right)} {\sqrt{\sum_{i=1}^{d} \sum_{j=1}^{n_i} a_{i,j}^2}} .\]
In fact, it is clear that the value of the maximization does not change if we optimize over tuples~$(\widetilde a_{i,j})$  of \emph{rational} numbers satisfying
\begin{equation}\label{eq:rational system}
\begin{aligned}
\forall \: 1 \le i \le d \colon \quad &\sum_{j=1}^{n_i} \widetilde a_{i,j} = 0 \quad \text{ and } \\
\forall \: (j_1,\ldots, j_{d}) \in S \colon \quad &\sum_{i=1}^{d} \widetilde a_{i, j_i} \geq 1
\end{aligned}
\end{equation}
instead (since the objective function is invariant under rescaling all components $a_{i,j}$).
We now use~\cite[Theorem 10.1]{schrijver1998theory}, which ensures that there exist a rational solution of \cref{eq:rational system} of bitsize $O(n_1 \cdots n_d \cdot \log(n_1 \cdots n_d))$.
But this implies that
\[
  \frac {\min_{(j_1,\dots,j_{d}) \in S}\left(\sum_{i=1}^{d} \widetilde a_{i,j_i}\right)} {\sqrt{\sum_{i=1}^{d} \sum_{j=1}^{n_i} \widetilde a_{i,j}^2}}
= \exp(-O(n_1\cdots n_d \cdot \log(n_1 \cdots n_d)))
= \exp(-O(n \log n)),
\]
completing the proof.
\end{proof}

We now prove a slight variation on~\cite[Theorem 4.6]{BCCGNSU17}:

\slicerank*
\begin{proof}
We adapt the proof of~\cite[Theorem 4.6]{BCCGNSU17} to $d$ tensor factors and our definition of the instability (\cref{def:ins}).
Since $\sr(X)<m$ we can find vectors $V^{(i,k)}\in\C^m$
and tensors $Y^{(i,k)} \in \Ten_m^{d-1}$ for $i\in[d]$ and $k\in [r_i]$ such that we can write
$X$ as a sum of $\sum_{i=1}^d r_i = \sr(X)$ many slice rank one tensors:
\[
  X_{j_1,\dots,j_d} = \sum_{i=1}^d \sum_{k=1}^{r_i} V^{(i,k)}_{j_i} Y^{(i,k)}_{j_1,\dots,j_{i-1},j_{i+1},\dots,j_d}.
\]
The vectors $V^{(i,k)}$, for a fixed $i$ and varying $k$, are linearly independent (otherwise the above decomposition would not be minimal).
By performing a basis transformation $B=(B_1,\dots,B_d)$, we can ensure that $V^{(i,k)} = e_k$, where $e_1,\ldots, e_m$ is the standard basis of $\C^m$.
It is clear that
\begin{equation}\label{eq:support for slice rank degenerate}
(j_1,\dots,j_d) \in \supp(B \cdot X) \Longrightarrow \exists i\in[d]: j_i \leq r_i.
\end{equation}
Now we define
\[ a_{i,j} := \begin{cases}
m - r_i & \text{ if } 1 \leq j \leq r_i, \\
-r_i & \text{ if } r_i < j \leq m
\end{cases}
\]
for $i\in[d]$, $j\in[m]$.
Then $(a_{i,j}) \in \Gamma$ and so
\begin{align*}
&\quad \ins(X)
\geq \defi(\supp(B \cdot X))
\geq \min_{(j_1,\dots,j_d) \in \supp(B \cdot X)} \frac {\sum_{i=1}^d a_{i,j_i}} {\sqrt{\sum_{i=1}^d \sum_{j=1}^m a_{i,j}^2}} \\
&\geq \min_{i\in[d]} \frac {(m - r_i) - \sum_{l=1, l\neq i}^d r_l} {\sqrt{\sum_{i=1}^d \sum_{j=1}^m a_{i,j}^2}}
= \frac {m - \sr(X)} {\sqrt{\sum_{i=1}^d \sum_{j=1}^m a_{i,j}^2}}
\geq \frac 1 {\sqrt{\sum_{i=1}^d \sum_{j=1}^m a_{i,j}^2}}
\geq \frac 1 {\sqrt{d m^3}}
\end{align*}
(the first inequality uses \cref{lem:instability S vs U} and the third step is due to \cref{eq:support for slice rank degenerate}).
\end{proof}

We finally prove \cref{le:asympt-stab}.

\asymptstab*
\begin{proof}
If $X$ is in the null cone, then $0 \in \overline{\SL(m)^d \cdot X}$, so also $0 \in \overline{\SL(m^k)^d \cdot X^{\ot k}}$
(just act on the first copy of $X$, say), hence $X^{\ot k}$ is likewise in the null cone.

Conversely, suppose that $X$ is \emph{not} in the null cone.
By \cref{thm:kempf ness}, there exists a tensor $Y\in\overline{\SL(m)^d \cdot X}$ such that $\ds(Y)=0$.
But then $Y^{\ot k}\in\overline{(\SL(m)^d)^{\ot k} \cdot X^{\ot k}}\subseteq \overline{\SL(m^k)^d \cdot X^{\ot k}}$
as well is a tensor with $\ds(Y^{\ot k})=0$;
so we conclude that, by the other direction of \cref{thm:kempf ness}, $X^{\ot k}$ is not in the null cone either.
\end{proof}

\addcontentsline{toc}{section}{References}
\bibliographystyle{alpha}
\newcommand{\etalchar}[1]{$^{#1}$}

\end{document}